\def\breakequations{1} % 1= break equations to fit into double columns

\documentclass[twocolumn]{IEEEtran}

\usepackage{url}
\usepackage{latexsym}
\usepackage{amscd,amsmath}
\usepackage{amsfonts}
\usepackage{color}
\usepackage{float}
\usepackage{longtable}
\usepackage{graphicx}
\usepackage{xspace}

\newtheorem{thm}{Theorem}[section]
\newtheorem{lem}[thm]{Lemma}
\newtheorem{cor}[thm]{Corollary}
\newtheorem{propo}[thm]{Proposition}
\newtheorem{clm}[thm]{Claim}
\newtheorem{defn}[thm]{Definition}
\newtheorem{assm}[thm]{Assumption}
\newtheorem{rem}[thm]{Remark}
\newtheorem{obs}[thm]{Observation}
\newtheorem{egs}[thm]{Example}

\newtheorem{fct}[thm]{Fact}
\newtheorem{cons}[thm]{Construction}
\newtheorem{nte}[thm]{Note}

\newenvironment{theorem}{\begin{thm}}{\end{thm}}
\newenvironment{lemma}{\begin{lem}}{\end{lem}}

\newenvironment{definition}{\begin{defn}}{\end{defn}}

\newenvironment{remark}{\begin{rem}\begin{em}}{\end{em}\end{rem}}

\newcommand{\secref}[1]{Section~\ref{#1}}

\newcommand{\thref}[1]{Theorem~\ref{#1}}

\newcommand{\lemref}[1]{Lemma~\ref{#1}}

\newcommand{\remref}[1]{Remark~\ref{#1}}

\renewcommand{\eqref}[1]{\mbox{Equation~(\ref{#1})}}

\newcommand{\calA}{{\cal A}}
\newcommand{\calB}{{\cal B}}
\newcommand{\calC}{{\cal C}}
\newcommand{\calD}{{\cal D}}

\newcommand{\calG}{{\cal G}}

\newcommand{\calT}{{\cal T}}

\newcommand{\calV}{{\cal V}}
\newcommand{\calW}{{\cal W}}
\newcommand{\calX}{{\cal X}}
\newcommand{\calY}{{\cal Y}}

\newcommand{\N}{{{\mathbb N}}}

\def\bits{\{0,1\}}

\newcommand{\Prob}[1]{{\Pr\left[\,{#1}\,\right]}}
\newcommand{\Probsub}[2]{{\Pr_{#1}\left[\,{#2}\,\right]}}

\def\getsr{\stackrel{{\scriptscriptstyle\$}}{\leftarrow}}

\def\suchthatt{\: :\:}
\newcommand{\suchthat}{{\mbox{s.t.\ }}}

\newcommand{\Colon}{{:\;\;}}

\DeclareMathOperator{\Ext}{\mathsf{Ext}}
\DeclareMathOperator{\Rec}{\mathsf{Rec}}
\DeclareMathOperator{\HD}{\mathsf{HD}}

\DeclareMathOperator{\im}{Im}
\DeclareMathOperator{\Samp}{\mathsf{Samp}}
\DeclareMathOperator{\Bad}{\mathcal{BAD}}
\newcommand{\alice}{\ensuremath{\mathsf{Alice}}\xspace}
\newcommand{\bob}{\ensuremath{\mathsf{Bob}}\xspace}
\newcommand{\view}[3]{\ensuremath{\mathsf{view}_{#1}^{#2}(#3)}\xspace}
\newcommand{\trans}[2]{\ensuremath{\mathsf{trans}^{#1}(#2)}\xspace}
\newcommand{\viewbob}[2]{\view{\bob}{#1}{#2}}
\newcommand{\viewalice}[2]{\view{\alice}{#1}{#2}}
\def\CP{\mathsf{CP}}
\def\OP{\mathsf{OP}}
\def\OT{\mathsf{OT}}
\def\IH{\mathsf{IH}}

\def\lcor{\lambda_{\mathsf{C}}}
\def\lhid{\lambda_{\mathsf{H}}}
\def\lbin{\lambda_{\mathsf{B}}}
\def\lalice{\lambda_{\mathsf{A}}}
\def\lbob{\lambda_{\mathsf{B}}}
\def\test{\mathsf{test}}

\begin{document}

\title{Commitment and Oblivious Transfer in the Bounded Storage Model with Errors
\thanks{A conference version of part of this work appeared at the proceedings of ISIT 2014~\cite{ISIT:DowLacNas14}.}
}

\author{Rafael~Dowsley, Felipe Lacerda and Anderson~C.~A.~Nascimento
\thanks{Rafael~Dowsley and Felipe Lacerda are with the Department of Computer Science, Aarhus University, \r{A}bogade 34, 8200, Aarhus N, Denmark. Emails: rafael@cs.au.dk, lacerda@cs.au.dk.}
\thanks{Anderson~C.~A.~Nascimento is with the Institute of Technology, University of Washington Tacoma, 1900 Commerce Street, Tacoma, WA 98402-3100, USA. E-mail: andclay@uw.edu.}
}

\maketitle

\begin{abstract}
The bounded storage model restricts the memory of an adversary in a cryptographic protocol, rather than restricting its computational power, making information theoretically secure protocols feasible. 
We present the first protocols for commitment and oblivious transfer in 
the bounded storage model with errors, i.e., the model where the public random sources 
available to the two parties are not exactly the same, but instead are only required 
to have a small Hamming distance between themselves. Commitment and 
oblivious transfer protocols were known previously only for the error-free variant of the bounded 
storage model, which is harder to realize.  
\end{abstract}


\IEEEpeerreviewmaketitle

\section{Introduction}

\textit{Commitment schemes} are fundamental building blocks of modern 
cryptography. They are important in the construction of protocols such as
identification protocols~\cite{C:FiaSha86}, contract signing~\cite{EveGolLem:1985},
zero-knowledge proofs~\cite{GolMicWig91}, coin flipping over the
phone~\cite{Blum:1983}, and more generally in two- and multi-party computation
protocols~\cite{STOC:GolMicWig87,C:ChaDamVan87}.
A commitment scheme is a two-stage protocol between two parties, Alice and
Bob. First they execute the \textit{commit stage}, in which Alice chooses a
value $v$ as input and commits to it. Later, they execute the \textit{open
  stage}, in which Alice reveals $v$ to Bob. For the protocol to be secure, it
must satisfy two conditions:
the \textit{hiding property}, which means that Bob cannot learn any
information about $v$ before the open stage, and the \textit{binding property},
which means that after the commit phase, Alice cannot change $v$ without 
that being detected by Bob.

\textit{Oblivious transfer} (OT) is another essential primitive for two- and multi-party computation. It is a two-party protocol in which Alice 
inputs two strings $s_0$ and $s_1$, and Bob inputs a bit $c$. Bob's 
output is the string $s_c$. The protocol is secure if Alice never 
learns the choice bit $c$ and Bob does not learn any information about 
$s_{1-c}$. OT is a fundamental building block for 
multi-party computation and can be used to realize \textit{any} secure 
two-party computation~\cite{STOC:Kilian88,C:IshPraSah08}. 

Unconditionally secure commitment and OT are
impossible in the setting where the parties only communicate through noiseless
channels (even if quantum channels are available~\cite{Mayers:1997}).
However, both of them are possible in the context of 
computational security (in which the adversaries are restricted to be polynomial-time 
Turing machines), as long as computational hardness assumptions are made. 
Commitment can be obtained using generic assumptions such as the
existence of pseudorandom generator~\cite{JC:Naor91} or (more efficiently) 
assuming the hardness of various specific computational 
problems~\cite{C:Even81,Blum:1983,C:Pedersen91}.
OT can be obtained from dense trapdoor 
permutations~\cite{TCC:Haitner04} (which is conjectured to be stronger
than pseudorandom generators) or assuming the hardness of 
many specific computational problems~\cite{Rabin:1981,C:BelMic89,
EC:Kalai05,C:PeiVaiWat08,ICITS:DGMN08,IEICE:DGMN12,CANS:DavDowNas14,BDDMN17}.

Physical assumptions, such
as the existence of noisy channels, enable one to obtain unconditional security for commitment and OT
protocols. In this scenario the problem was studied
from both the theoretical~\cite{EC:DamKilSal99,WinNasIma:2003,ImaMorNas:2006,AhlCsi:2007,NasWin:2008,PDMN:11,DowNas14} 
and the efficient protocol designing~\cite{CreKil:1988,EC:Crepeau97,
SteWol:2002,SCN:CreMorWol04} points of view.

In this paper, we consider a different setting, the so called \textit{bounded storage 
model} (BSM)~\cite{JC:Maurer92a}, in which the adversary is assumed to have bounded memory. 

\subsection{The Bounded Storage Model}

The bounded storage model assumes that both parties have access to
a public random string, and that a dishonest party cannot store the whole
string. This string can be obtained from a natural source, from a trusted third
party, or, in some cases even generated by one of the parties.

A variety of cryptographic tasks can be implemented in the bounded storage
model. Cachin and Maurer~\cite{C:CacMau97} proposed a key agreement
protocol in the bounded storage model in which the parties have a small
pre-shared key, and use it to select bits from a public random source of size
$n$. Key agreement in this setting is always possible if the
pre-shared key has size proportional to $\log n$, as long as the adversary has
bounded memory. They also
proposed a protocol for key agreement by public discussion (that is, without a
pre-shared key) that requires $\sqrt{kn}$ (where $k$ is the key length) samples from the random
source and is thus less practical. Later, Dziembowski and Maurer~\cite{DM08} showed that 
this protocol is optimal, in the sense that one cannot have key
agreement by public discussion in the bounded storage model with less than
$O(\sqrt n)$ samples.

The first OT protocol in the bounded storage model was 
introduced by Cachin et al.~\cite{FOCS:CacCreMar98}. Ding~\cite{C:Ding01} and
Hong et al.~\cite{AC:HonChaRyu02} proposed improvements in a slightly different model. Ding et al.~\cite{TCC:DHRS04} obtained
the first constant-round protocol.

Shikata and Yamanaka~\cite{IMA:ShiYam11} and independently
Alves~\cite{Alves:2010} studied the problem of commitment in the
bounded storage model and provided protocols for it.

Unfortunately the bounded storage model assumes that there
exists a random source that can be reliably broadcast to all parties, without
errors in the transmission, and this is hard to realize in practice.

Consider a scenario where a satellite broadcasts a very large random string to be used in protocols in the bounded storage model. In his Ph.D. thesis, Ding \cite{Ding01} made an analysis of the practicality of this scenario, showing that an antenna with a surface area of $10m^2$ can be used to receive random bits from a geo-stationary satellite at rates up to 50 Gbps.  Ding's analysis did not consider the fact that errors are introduced in the string and that an adversary might be able to jam signal received by a legitimate party.  Our goal with this work is to study two-party protocols under more 
realistic assumptions.

\subsection{Our contribution}

In this work, we consider a more general variant of the BSM, in which errors are introduced in the public random source in arbitrary positions. 
This setting captures the situation in which the source is partially controlled
by an adversary, and also the situation in which there are errors due to 
noise in the channel. It is only assumed that the fraction of errors, relative to
the length of the public string, is not too large. Ding previously studied this model~\cite{TCC:Ding05}  in the context of secret key extension protocols. These protocols can  be modified, at the cost of an efficiency loss, to handle the case of key agreement, when no pre-shared key exists.
 He defined a general paradigm for BSM randomness extraction schemes and also showed how to
incorporate error correction in key agreement extension by using fuzzy
extractors~\cite{EC:DodReySmi04}.   

We give a brief introduction to the model and its notation in order to state our results.  A transmission 
phase is executed prior to the realization of the protocols' main part. In this phase, Alice has access to a sample $x \in \bits^n$ from an 
$\alpha n$-source $X$ (a source with min-entropy at least $\alpha n$), where $0 < \alpha < 1$, and the Bob to $\widetilde{x} \in \bits^n$
such that $\HD(x,\widetilde{x}) \leq \delta n$, where $\HD(\cdot)$ represents the hamming distance. We assume that an adversary (controlling either Alice or Bob) has complete control on where to insert the differences between the strings $x$ and $\widetilde{x}$, thus capturing both the situation where the source is noisy and  the situation where an adversary controls part of the source. 

We propose the first protocols for bit commitment and OT in the BSM with errors, thus extending the results of \cite{TCC:Ding05} to the case of two-party secure protocols. We show that the techniques introduced by \cite{TCC:Ding05} originally in the context of key extension give us efficient protocols for implementing OT.  Our OT protocol assumes a memory bounded Bob (i.e., he is able to store at most $\gamma n$ bits for 
$\gamma < \alpha$), but no limitation is put on Alice's memory. It works based on an efficient linear error correcting code proposed in \cite{guruswami2002near} with rate $\beta$ and achieving the Zyablov bound \cite{zyablov71}. We show that as long as $\beta > 1 -\alpha -\gamma$ the protocol works for noise levels $\delta$ as severe as (approximately) $$\max_{\beta<\widetilde{\beta}<1}\frac{(1-\widetilde{\beta})y}{2},$$ where $y$ is the unique value in $[0,1/2]$ so that $h(y)=1-\beta/\widetilde{\beta}$ and $h(\cdot)$ is the binary entropy function.  In case a random linear error correcting code is used an improved noise level can be tolerated $$h(2\delta) < \alpha - \gamma.$$ This improvement in the resilience comes at the price of making the protocol inefficient from a computational complexity point of view, given the intractability of decoding random linear codes. 

This OT protocol immediately gives us a commitment scheme. However, using OT for obtaining commitment is not a desirable solution. The communication, round and computational complexities of OT protocols are usually much higher than the ones for commitment schemes. Moreover, it could be the case that commitment protocols  work for different ranges of noise $\delta$. 

We propose a direct construction of a commitment protocol that does not rely on the framework proposed by Ding~\cite{TCC:Ding05}, does not use error correcting codes at all, implements \emph{string} commitment and has only one message from Bob to Alice. Again, we assume that Bob has limited memory. No limitations are imposed on Alice whatsoever. The protocol is very efficient and simple and works for $$h(\delta)<\frac{\alpha-\gamma}{2}.$$ We then show that it is possible to obtain a protocol that works for a much larger range of noise  $$h(\delta)<\alpha-\gamma$$ at the cost of having one additional message in each direction and by using a family of $4k$-universal hash functions. Finally, we show that the use of families of $4k$-universal hash functions can be avoided by imposing a memory bound on Alice, instead of Bob.  We note that being able to implement protocols in the memory bounded by bounding any of the parties is an important matter, particularly when one of the parties is much more powerful than the other. This protocol is based on the interactive hashing protocol of \cite{TCC:DHRS04} and also works for $$h(\delta)<\alpha-\gamma,$$ but has extra rounds of communication and implements \emph{bit} rather than string commitment. 

The techniques we use in our results are standard in the field: extractors, error-correcting codes, typicality tests, sampling, etc. However, to the best of our knowledge, this is the first time that these techniques are combined to obtain commitment and OT protocols in the memory bounded model with errors. Moreover, the study of how much adversarial noise can be tolerated in this model and its relation to round complexity is also original, as far as we know.  Interestingly, the noise levels tolerated by our protocols are different for OT and commitment schemes. This contrasts sharply with the noiseless situation where either one has every possible secure two-party computation or nothing at all.

\subsection{Overview}

In \secref{sec:prelim} we present the main tools used in our protocols. \secref{sec:model} explains the security model. 
Our commitment protocols are introduced in \secref{sec:commit} and the oblivious transfer in
\secref{sec:protocol-ot}. A conference version of this work appeared at the proceedings of ISIT 2014~\cite{ISIT:DowLacNas14} and only covered the case of OT. 
In this full version a more detailed presentation of the case of OT is presented and the case of commitment is entirely new; the other sections are also extended accordingly.

\section{Preliminaries}\label{sec:prelim}

We use calligraphic letters for denoting domains of random variables and other sets, upper case letters for random variables and 
lower case letters for realizations of the random variables.
We deal solely with discrete random variables. The probability mass function of a random variable $X$ will be denoted by 
$P_X$. The set $\{1,\dots,n\}$ will be written as $[n]$. If $x = (x_1, \dots, x_n)$ 
is a sequence and $S = \{s_1, \dots, s_t\} \subseteq [n]$, $x^S$ denotes 
the sequence $(x_{s_1}, \dots, x_{s_t})$. $u \getsr U$ denotes 
that $u$ is drawn from the uniform distribution over the set $U$ and
$U_r$ is the uniformly-distributed $r$-bit random variable. $y \getsr \mathcal{F}(x)$
denotes the act of running the probabilistic algorithm $\mathcal{F}$ with
input $x$ and obtaining the output $y$. $y \gets \mathcal{F}(x)$ is 
similarly used for deterministic algorithms.

If $x$ and $y$ are strings, $\HD(x,y)$ denotes their Hamming
distance (that is, the number of positions in which they differ) and $x
\oplus y$ their bitwise exclusive or. Let $\log x$ denote the logarithm 
of $x$ in base 2. The binary entropy
function is denoted by $h$: for $0 \leq x \leq 1$, $h(x) = -x \log x - (1-x) \log (1-x)$. By
convention, $0 \log 0 = 0$. $H(X)$ denotes the entropy of $X$ and
$I(X;Y)$ the mutual information between $X$ and $Y$.

The \textit{statistical distance} is a measure of the distance between two
probability distributions. Here, we state its definition for the case of discrete probabilities. 

\begin{definition}[Statistical distance]
The \textnormal{statistical distance} $\|P_X-P_Y\|$ between two probability
mass functions $P_X,P_Y$ over an alphabet $\mathcal{X}$ is defined as
\begin{displaymath}
\|P_X-P_Y\| = \max_{A \subseteq \mathcal{X}} \left|\sum_{x \in A} P_X(x) - P_Y(x)\right|.
\end{displaymath}
We say $P_X$ and $P_Y$ are $\varepsilon$-close if $\|P_X - P_Y\| \leq \varepsilon$.
\end{definition}

%\subsection{Entropy Measures}

The main entropy measure in this work is the \textit{min-entropy}.

\begin{definition}[Min-entropy]
  Let $P_{XY}$ be a probability mass function over $\mathcal{X} \times
  \mathcal{Y}$. The \textnormal{min-entropy of $X$}, denoted by $H_{\infty}(X)$,
  and the \textnormal{conditional min-entropy of $X$ given $Y$}, denoted by $H_{\infty}(X
  | Y)$, are respectively defined as
  \begin{align*}
    H_{\infty}(X) &= \min_{x \in \mathcal{X}} (-\log P_X(x)) \text{ and} \\
    H_{\infty}(X | Y) &= \min_{y \in \mathcal{Y}} \min_{x \in \mathcal{X}}
    (-\log P_{X | Y=y}(x)).
  \end{align*}

$X$ is called a \textnormal{$k$-source} if $H_\infty(X) \geq k$.
\end{definition}

The conditional min-entropy $H_{\infty}(X | Y)$ measures the extractable
private randomness from the variable $X$, given the correlated random variable $Y$
possessed by an adversary.
The min-entropy has the problem of being sensitive to small changes in
the probability mass function and due to this fact the notion of \textit{smooth} 
min-entropy~\cite{AC:RenWol05} will be used (please refer to the Appendix for a longer discussion about 
the advantages of the smooth variants for cryptographic applications).

\begin{definition}[Smooth min-entropy]
Let $\varepsilon > 0$ and $P_{XY}$ be a probability mass functions. The 
\textnormal{$\varepsilon$-smooth min-entropy of $X$ given $Y$} is 
defined by
\begin{displaymath}
H_\infty^\varepsilon(X|Y)=\max\limits_{X'Y':\|P_{X'Y'}-P_{XY}\|
\leq\varepsilon}H_\infty(X'|Y').
\end{displaymath}
\end{definition}

Intuitively, the smooth min-entropy is the maximum min-entropy in the neighborhood
of the probability mass function. Similarly, we also define the max-entropy and its smooth version.

\begin{definition}[(Smooth) Max-entropy] The max-entropy is defined as 
\[
H_0(X)=\log|\{x\in X|P_{X}(x)>0\}|
\]
and its conditional version is given by
\[
H_0(X|Y)=\max\limits_{y}H_0(X|Y=y).
\]
The smooth variants are defined as
\[
H_0^\varepsilon(X)=\min\limits_{X':\|P_{X'}-P_{X}\|
\leq\varepsilon} H_0(X') \text{ and}
\]
\[
H_0^\varepsilon(X|Y)=\min\limits_{X'Y':\|P_{X'Y'}-P_{XY}\|
\leq\varepsilon} H_0(X'|Y').
\]
\end{definition}

The following inequalities are smooth min-entropy
analogues of the chain rule for conditional Shannon
entropy.

\begin{lemma}[\cite{AC:RenWol05}]
  \label{lem:chain}
  Let $\varepsilon,\varepsilon',\varepsilon'' > 0$ and $P_{XYZ}$ be a tripartite probability
  mass function. Then
%%%%%%%%%%%%%%%%%%%%%%%%%%%%%%%
% In case of changes, change both
%%%%%%%%%%%%%%%%%%%%%%%%%%%%%%%
  \ifnum\breakequations=1
  	\begin{align*}
    		H_{\infty}^{\varepsilon + \varepsilon'}(X, Y | Z) &\geq
    		H_{\infty}^{\varepsilon}(X | Y, Z) + H_{\infty}^{\varepsilon'}(Y | Z) \text{ and}\\
    		H_{\infty}^{\varepsilon}(X, Y | Z) &<
    		H_{\infty}^{\varepsilon+\varepsilon' + \varepsilon''}(X | Y, Z) +
    		H_0^{\varepsilon''}(Y | Z) \\ 
    		& \qquad + \log(1/\varepsilon'). 
  	\end{align*}
   \else
   	\begin{align*}
    		H_{\infty}^{\varepsilon + \varepsilon'}(X, Y | Z) &\geq
    		H_{\infty}^{\varepsilon}(X | Y, Z) + H_{\infty}^{\varepsilon'}(Y | Z) \text{ and}\\
    		H_{\infty}^{\varepsilon}(X, Y | Z) &<
    		H_{\infty}^{\varepsilon+\varepsilon' + \varepsilon''}(X | Y, Z) +
    		H_0^{\varepsilon''}(Y | Z) 
    		+ \log(1/\varepsilon'). 
  	\end{align*}
   \fi
\end{lemma}

The notion of \textit{min-entropy rate} and a few results regarding its preservation
will be used subsequently.

\begin{definition}[Min-entropy rate]
Let $X$ be a random variable with an alphabet $\mathcal{X}$, $Y$ be an
arbitrary random variable, and $\varepsilon \geq 0$. The
\textnormal{min-entropy rate} $R_{\infty}^\varepsilon(X | Y)$ 
is defined as
\begin{displaymath}
R_{\infty}^\varepsilon(X | Y) = \frac{H_{\infty}^\varepsilon(X | Y)}{\log |\mathcal{X}|}.
\end{displaymath}
\end{definition}

The following lemma follows by using Lemma 3.16 in \cite{JC:DHRS07} as a first step 
and some easy extra steps to obtain an inequality in terms of the min-entropy rate alone. It 
says that a source with high min-entropy also has high min-entropy when
conditioned on a correlated short string. This lemma is what makes the bounded
storage assumption useful: it implies that a memory bounded
adversary has limited information about a string sampled from the public random string.

\begin{lemma}
\label{lem:gr}
Let $X \in \{0,1\}^n$ such that $R_{\infty}^{\varepsilon}(X) \geq \rho$ and $Y$ be 
a random variable over $\{0,1\}^{\phi n}$. Fix $\varepsilon' >0$. Then
\begin{displaymath}
R_{\infty}^{\varepsilon'+\sqrt{8 \varepsilon}}(X | Y) \geq \rho - \phi -
\frac{1 + \log(1/\varepsilon')}{n}.
\end{displaymath}
\end{lemma}

\begin{IEEEproof}
Let $\psi = \rho - \phi - \frac{1+\log(1/\varepsilon')}{n}$. By lemma~3.16 in~\cite{JC:DHRS07} we have that if
$R_{\infty}^{\varepsilon}(X) \geq \rho$ then

\begin{displaymath}
\Pr_{y \getsr Y}\left[R_{\infty}^{\sqrt{2 \varepsilon}}(X | Y = y)
\geq \psi \right] \geq 
1 - \varepsilon' - \sqrt{2 \varepsilon}.
\end{displaymath}

To get the desired result, let $\mathcal{G} = \{y \in \mathcal{Y} | R_{\infty}^{\sqrt{2 \varepsilon}}(X
| Y = y) \geq \psi\}$ and $P_{XY}$ be the joint probability distribution
of $X$ and $Y$. Let $P'_{XY}$ be the distribution that is $\sqrt{2\varepsilon}$-close
to $P_{XY}$ and is such that $P'(X = x | Y = y) \leq 2^{-\psi n}$ for any
$x \in \mathcal{X}, y \in \mathcal{G}$. Let $P''_{XY}$ be obtained by letting 
$P''(X | Y = y)=P'(X | Y = y)$ for $y \in \mathcal{G}$ and
defining $P''(X =x | Y = y) = 2^{-n}$ for any $x \in \mathcal{X}$, $y \notin \mathcal{G}$. 
As $\Prob{\cal G} \geq 1 - \varepsilon' - \sqrt{2 \varepsilon}$, it holds that
$\|P''_{XY} - P'_{XY}\| \leq \varepsilon' + \sqrt{2 \varepsilon}$ and
so $\|P''_{XY} - P_{XY}\| \leq \varepsilon' + 2\sqrt{2 \varepsilon}$.
Since $P''(X=x | Y=y) \leq 2^{-\psi n}$ for every $x \in \mathcal{X}, 
y \in \mathcal{Y}$, the lemma follows.
\end{IEEEproof}

\subsection{Averaging Samplers and Randomness Extractors}\label{sec:ext}

The sample-then-extract paradigm is usually employed  in the bounded storage model - first some 
positions of the source are sampled and then an extractor is applied on
these positions. Note that due to the assumption that it is infeasible to
store the whole source string (the memory bound), it is not possible to apply an extractor 
to the complete string at once, the extractor needs to be locally computable~\cite{JC:Vadhan04}.
In this context, \textit{averaging samplers}~\cite{FOCS:BelRom94,
IPL:CanEveGol95, Zuckerman97} are a fundamental 
tool. Intuitively, averaging samplers produce samples such that
the average value of any function applied to the sampled string is roughly the
same as the average when taken over the original string.

\begin{definition}[Averaging sampler]
A function $\Samp\Colon \allowbreak \bits^r \to [n]^t$ is an
\textnormal{$(\mu,\nu,\varepsilon)$-averaging sampler} if for every function
$f\Colon [n] \to [0,1]$ with average $\frac{\sum_{i=1}^n f(i)}{n} \geq \mu$ it holds that
\begin{equation*}\label{eq:avgdef}
\Probsub{\mathcal{S} \getsr \Samp(U_r)}{\frac{1}{t} \sum_{i \in \mathcal{S}} f(i) \leq \mu - \nu} \leq \varepsilon.
\end{equation*}
\end{definition}

Averaging samplers enjoy several useful properties. Particularly important to this 
work is the fact that averaging samplers roughly preserve the \textit{min-entropy
rate}.

\begin{lemma}[\cite{JC:Vadhan04}]\label{lem:vadhan_minentropy}
Let $X \in \{0,1\}^n$ be such that $R_{\infty}(X | Y) \geq \rho$.
Let $\tau$ be such that $1\geq \rho \geq 3\tau>0$ and 
$\Samp \Colon \bits^r \to [n]^t$ be an
$(\mu,\nu,\varepsilon)$-averaging sampler with distinct samples for 
$\mu = (\rho - 2\tau)/\log(1/\tau)$ and $\nu = \tau/\log(1/\tau)$. Then
for $\mathcal{S} \getsr \Samp(U_r)$
\begin{displaymath}
R_{\infty}^{\varepsilon'}(X^{\mathcal{S}} | \mathcal{S}, Y) \geq \rho -
3\tau,
\end{displaymath}
where $\varepsilon' = \varepsilon + 2^{-\Omega(\tau n)}$.
\end{lemma}

For $t < n$, the uniform distribution over subsets of $[n]$ of size $t$ is an
averaging sampler, also called the $(n,t)$-\textit{random subset sampler}.

\begin{lemma}\label{lem:random-subset}
Let $0 < t< n$. For any $\mu,\nu > 0$, the $(n,t)$-random subset sampler is a 
$(\mu,\nu,e^{-t \nu^2/2})$-averaging sampler.
\end{lemma}

\begin{IEEEproof}
It is just a restatement of Lemma~5.5 in~\cite{BabHay:2005}.
\end{IEEEproof}

A \textit{randomness extractor} is a function that takes a string with high
min-entropy as an input and outputs a string that is close (in the statistical
distance sense) to a uniformly distributed string.

\begin{definition}[Strong extractor]
A function $\Ext \Colon \{0,1\}^n \times \{0,1\}^r \to \{0,1\}^m$ is a 
\textnormal{$(k,\varepsilon)$-strong extractor} if for every $k$-source
$X$, we have
\begin{displaymath}
\|P_{\Ext(X,U_r),U_r} - P_{U_m,U_r}\| \leq \varepsilon.
\end{displaymath}
\end{definition}

The following lemma specifies the parameters of an explicit strong extractor 
construction~\cite{Zuckerman97}.

\begin{lemma}[\cite{Zuckerman97}]\label{lem:ext}
Let $\rho, \psi>0$ be arbitrary constants. For every $n \in \N$ and every 
$\varepsilon>e^{-n/2^{O(\log^* n)}}$, there is an explicit construction of a
$(\rho n, \varepsilon)$-strong extractor $\Ext \Colon \{0,1\}^n \times \{0,1\}^r \to \{0,1\}^m$
with $m= (1-\psi)\rho n$ and $r=O(\log n + \log (1/\varepsilon))$.
\end{lemma}

The OT protocol presented in this work uses a variant of a strong extractor
called a \textit{fuzzy extractor}~\cite{EC:DodReySmi04}. Intuitively, fuzzy
extractors are noise-resilient extractors, that is, extractors such that the
extracted string can be reproduced by any party with a string that is close (in
the Hamming distance sense) to the original source.

\begin{definition}[Fuzzy extractor]
A pair of functions $\Ext \Colon \{0,1\}^n \times \{0,1\}^r \to \{0,1\}^m
\times \{0,1\}^q$, $\Rec \Colon \{0,1\}^n \times \{0,1\}^r \times \{0,1\}^q \to \{0,1\}^m$ is
an \textnormal{$(k, \varepsilon, \delta, \beta)$-fuzzy extractor} if:

\begin{itemize}
\item  For every $\kappa$-source $X \in \{0,1\}^\ell$, 
$(Y,Q) \gets \Ext(X,U_r)$. Then $\|P_{Y U_r Q} - U_m \times P_{U_r Q}\| \leq \varepsilon$.

\item For every $x, x' \in \{0,1\}^\ell$ such that $\HD(x, x') \leq \delta \ell$, 
let $r \getsr U_r$, $(y,q) \gets \Ext(x,r)$. Then it should hold that 
$\Prob{\Rec(x',r,q) = y} \geq 1-\beta$.
\end{itemize}
\end{definition}

Fuzzy extractors are a special case of \textit{one-way key-agreement
schemes}~\cite{C:HolRen05,EC:KanRey09}. Ultimately they are
equivalent to performing information reconciliation followed by privacy
amplification~\cite{EC:RenWol04}. Since there is a restriction to 
close strings with respect to the Hamming distance, syndrome-based fuzzy
extractors can be used, as summarized in the following lemma
from Ding~\cite{TCC:Ding05}.

\begin{lemma}[\cite{TCC:Ding05}]\label{lem:fuzzy-extractor-construction}
Let $1 \geq \rho, \psi>0$ and $1/4 > \delta > 0$ be arbitrary constants. There is a
constant $\beta$, depending on $\delta$, such that for every sufficiently large $n
\in \N$, and every $\varepsilon>e^{-n/2^{O(\log^* n)}}$, there exists an explicit
construction of a $(\rho n, \varepsilon, \delta, 0)$-fuzzy extractor $(\Ext, \Rec)$, where
$\Ext$ is of the form $\Ext \Colon \{0,1\}^n \times \{0,1\}^r \to \{0,1\}^m \times \{0,1\}^p$
with
\begin{align*}
  m &= (1-\psi)\rho n, \\
 r &= O\left(\log n + \log \frac{1}{\varepsilon}\right), \\
 p &\leq \frac{1-\beta}{(1-\psi) \rho} m.
\end{align*}
\end{lemma}

\begin{rem}
  \label{rem:fuzzy-extractor}
  The parameters $\beta,\delta$ refer to the error-correcting code used in the
  construction, specifically, a code of size $n$ with rate $\beta$ that can correct
  $\delta n$ errors. It is known~\cite{varshamov1957estimate} that, for a given
  $\upsilon$ with $0 < \upsilon < 1/2$ and $0 \leq \mu \leq 1
  - h(\upsilon)$, there exists a random linear code with minimum distance $\upsilon
  n$ and $\beta \geq 1 - h(\upsilon) - \mu$ (i.e., it matches the
  Gilbert-Varshamov bound). However this construction has no known efficient decoding. We can instead use the concatenated solution in Theorem 4
  of~\cite{guruswami2002near}, which achieves the Zyablov bound \cite{zyablov71}. The
  construction provides a code with linear-time encoding and decoding such that,
  for a given $0 < \beta < 1$ and $\mu > 0$, can correct $\delta n$ errors, where
  \begin{equation*}
    \label{eq:1}
    \delta \geq \max_{\beta < \widetilde{\beta} < 1} \frac{(1 - \widetilde{\beta} - \mu) y}{2}
  \end{equation*}
  and $y$ is the unique number in $[0,1/2]$ with $h(y) = 1-\beta/\widetilde{\beta}$.
\end{rem}

\subsection{Interactive Hashing and Binary Encoding of Subsets}\label{sec:ih}

Interactive hashing was initially introduced in the context of computationally secure cryptography \cite{OstVenYun:1993}, but was later
generalized to the information-theoretic setting, and is particularly useful in
the context of designing oblivious transfer \cite{FOCS:CacCreMar98,TCC:DHRS04,EC:CreSav06,Savvides:2007,PDMN:11} and commitment
protocols \cite{IMA:ShiYam11} with unconditional security.
In this primitive Bob inputs a string $w\in\bits^m$ and both 
Alice and Bob receive as output two strings $w_0,w_1\in\bits^m$ such that $w_0 \neq w_1$. The correctness requirement 
is that one of the two output strings, $w_d$, should be equal to $w$. The security guarantee for Alice is that one of the strings should be 
effectively beyond the control of (a malicious) Bob. On the other hand, the security guarantee for Bob states that (a malicious) Alice should
not be able to learn $d$.

A variety of protocols for realizing interactive hashing have been
proposed~\cite{FOCS:CacCreMar98,TCC:DHRS04,JC:NOVY98}. In this
work interactive hashing is used as a black box since the security of our 
protocols does not depend on particular features of the
interactive hashing protocol used, but only on its security properties.

\begin{definition}[Interactive hashing]
\textit{Interactive hashing} is a protocol between Alice and Bob
in which only Bob has an input $w \in \{0,1\}^m$, and both
parties output $w_0, w_1 \in \{0,1\}^m$ such that
$w_d = w$ for some $d \in \{0,1\}$. The protocol is called an
\textnormal{$\eta$-uniform $(t,\theta)$-secure interactive hashing protocol} if:

\begin{enumerate}
\item If both parties are honest, then the random variable $W_{1-d}$ is close to completely
random, i.e., $W_{1-d}$ is $\eta$-close to the uniform distribution on 
the $2^m - 1$ strings different from $w_d$.

\item Alice's view (possibly a malicious Alice) of the protocol is independent of $d$. Let  $\viewalice{\IH}{W}$ be Alice's view of the protocol 
when the input is the random variable $W$. Then
$$\left\{\viewalice{\IH}{W} | W = W_0 \right\} = \left\{ \viewalice{\IH}{W} | W = W_1 \right\}.$$

\item For any $\calT \subset \{0,1\}^m$ such that $|\calT| \leq 2^t$, it should hold that 
after the protocol execution between an honest Alice and a possibly malicious Bob, $$\Prob{W_0,W_1 \in \calT} \leq \theta,$$ where the probability
is over the parties' randomness.
\end{enumerate}
\end{definition}

By allowing $W_{1-d}$ to be distributed only closely to uniform, this definition
is weaker than the one usually given in the
literature~\cite{Savvides:2007}. It is, however, enough to prove
security of our oblivious transfer protocol. This more general definition allows
for the possibility of using the constant-round protocol of Ding et al.~\cite{TCC:DHRS04}
for interactive hashing. 

\begin{lemma}[\cite{TCC:DHRS04}]
\label{lem:interactive-hashing-ding}
Let $t, m$ be positive integers such that $t \geq \log m + 2$. Then there
exists a four-message $(2^{-m})$-uniform $(t,2^{-(m-t)+O(\log m)})$-secure 
interactive hashing protocol.
\end{lemma}

The following lemma is a result by~\cite{JC:NOVY98}. It is $0$-uniform
(that is, $W_{1-d}$ is distributed uniformly), and achieves near-optimal
security~\cite{Savvides:2007}, but has the disadvantage of taking
$m-1$ rounds to execute.

\begin{lemma}[\cite{JC:NOVY98}]
\label{lem:interactive-hashing-ccm}
There exists a $0$-uniform $(t,a \cdot 2^{-(m-t)})$-secure 
interactive hashing protocol for some constant $a > 0$.
\end{lemma}

A secure interactive hashing scheme guarantees that one of the outputs is
random; however, in the oblivious transfer protocols, the two binary strings are not
used directly, but as encodings of subsets of sequences. Thus for the protocol to succeed, 
both outputs need to be valid encodings of subsets of $\binom{[n]}{\ell}$.
The original protocol of Cachin et al.~\cite{FOCS:CacCreMar98} for oblivious transfer
used an encoding scheme that has probability of success $1/2$, thus requiring 
that the protocol be repeated several times to guarantee correctness. 
Later, Ding et al.~\cite{TCC:DHRS04} proposed a ``dense'' encoding of 
subsets, ensuring that most $m$-bit strings are valid encodings. More 
precisely, they showed the following result. 

\begin{lemma}[\cite{TCC:DHRS04}]\label{lem:dense-encoding}
Let $\ell \leq n$, $m \geq \lceil \log \binom{n}{\ell} \rceil$, $t_m = \lfloor
2^m/\binom{n}{\ell} \rfloor$. Then there exists an injective mapping $F \colon
\binom{[n]}{\ell} \times [t_m] \to [2^m]$ with $|\im (F)| > 2^m - \binom{n}{\ell}$.
\end{lemma}

\subsection{Miscellaneous}

Universal hash functions were introduced by Carter and Wegman \cite{CarWeg79} and are very useful in cryptography.

\begin{definition}[$t$-universal hash functions]
A family of functions $G = \{g \colon \mathcal{H} \to \mathcal{L}\}$ is called
a \textit{family of $t$-universal hash functions} if for $g \getsr G$ and for any $x_1, \ldots, x_t \in \mathcal{H}$, 
the induced distribution on $(g(x_1), \ldots, g(x_t))$ is uniform over $\mathcal{L}^t$.
\end{definition}

For any $\mathcal{H} = \bits^h$ and $\mathcal{L} = \bits^\ell$, there exists a $t$-universal family of hash functions for which the function description has size $poly(h,t)$ bits, and the sampling and computing times are in $poly(h,t)$.

The following is a basic fact that follows from simple counting.

\begin{lemma}\label{lem:entropy-hamming}
Let $0 \leq \delta < 1/2$ and let $x, y \in \{0,1\}^n$ such that $\HD(x,y) \leq \delta n$ and 
$H_{\infty}(X) \geq \alpha n$ where $0 < \alpha < 1$. Then $H_{\infty}(Y) \geq (\alpha - h(\delta)) n$.
\end{lemma}

The next lemma shows that the restrictions of two tuples to random subsets of their positions
have relative Hamming distances that are close to the one between the entire tuples. 

\begin{lemma}\label{lem:alice-bob-random-subset}
Let $x,y \in \{0,1\}^n$, $\cal S$ be a random subset of $[n]$ of size $r$
and consider any $\nu \in [0,1]$. On one hand, if $\HD(x,y) \leq \delta n$,
then $\HD(x^{\mathcal{S}},y^{\mathcal{S}}) < (\delta+\nu) r$ except with 
probability $e^{-r \nu^2/2}$. On the other hand, if $\HD(x,y) \geq \delta n$,
then $\HD(x^{\mathcal{S}},y^{\mathcal{S}}) > (\delta-\nu) r$ except with 
probability $e^{-r \nu^2/2}$.
\end{lemma}

\begin{IEEEproof}
Lets begin with the first part of the Lemma.
By \lemref{lem:random-subset}, a random subset sampler is an 
$(\mu,\nu,e^{-r \nu^2/2})$-averaging sampler for any $\mu, \nu > 0$. 
Hence for any $f \colon [n] \to [0,1]$ with $\frac{1}{n}\sum_{i=1}^n f(i) \geq \mu$
\begin{equation} \label{eq:avg}
\Prob{\frac{1}{|\mathcal{S}|} \sum_{i \in \mathcal{S}} f(i) \leq \mu - \nu } \leq e^{-r \nu^2/2},
\end{equation}
Let
  \begin{displaymath}
    f(i) =
    \begin{cases}
      0, & \text{if $x_i \neq y_i$}, \\
      1, & \text{otherwise}.
    \end{cases}
\end{displaymath}
Fix $\mu = 1 - \delta$. Note that $\frac{1}{|\mathcal{S}|} \sum_{i \in \mathcal{S}} f(i) = 1 -
\frac{\HD(x^{\mathcal{S}},y^{\mathcal{S}})}{r}$ and 
$\frac{1}{n}\sum_{i=1}^n f(i) = 1 - \frac{\HD(x,y)}{n} \geq \mu$. Thus by \eqref{eq:avg}
\begin{eqnarray*}
e^{-r \nu^2/2} & \geq & \Prob{\frac{1}{|\mathcal{S}|} \sum_{i \in \mathcal{S}} f(i) \leq \mu - \nu }\\
&=& \Prob{1 - \frac{\HD(x^{\mathcal{S}},y^{\mathcal{S}})}{r} \leq 1 - \delta -\nu}\\
&=& \Prob{\HD(x^{\mathcal{S}},y^{\mathcal{S}}) \geq (\delta + \nu)r},
\end{eqnarray*}
which proves the first part of the Lemma.

The second part of the Lemma uses the same idea, but now the function $f$ is
\begin{displaymath}
    f(i) =
    \begin{cases}
      0, & \text{if $x_i = y_i$}, \\
      1, & \text{otherwise}.
    \end{cases}
\end{displaymath}

Fixing $\mu = \delta$ it holds that
$\frac{1}{|\mathcal{S}|} \sum_{i \in \mathcal{S}} f(i) = \frac{\HD(x^{\mathcal{S}},y^{\mathcal{S}})}{r}$ and 
$\frac{1}{n}\sum_{i=1}^n f(i) = \frac{\HD(x,y)}{n} \geq \mu$ and hence
\begin{eqnarray*}
e^{-r \nu^2/2} & \geq & \Prob{\frac{1}{|\mathcal{S}|} \sum_{i \in \mathcal{S}} f(i) \leq \mu - \nu }\\
&=& \Prob{\frac{\HD(x^{\mathcal{S}},y^{\mathcal{S}})}{r} \leq \delta -\nu}\\
&=& \Prob{\HD(x^{\mathcal{S}},y^{\mathcal{S}}) \leq (\delta - \nu)r},
\end{eqnarray*}
which finishes the proof of the lemma.
\end{IEEEproof}

The following statement of the birthday paradox is standard.

\begin{lemma}
  \label{lem:birthday}
  Let $\calA, \calB \subset [n]$, chosen independently at random, with $|\calA| = |\calB| = 2
  \sqrt{\ell n}$. Then
  \begin{displaymath}
    \Pr[|\calA \cap \calB| < \ell] < e^{-\ell/4}.
  \end{displaymath}
\end{lemma}

\begin{IEEEproof}
  See corollary 3 in~\cite{C:Ding01}.
\end{IEEEproof}

The following useful lemma will also be needed in the subsequent sections.

\begin{lemma}[\cite{alon2004probabilistic}]
  \label{lem:bound_binom}
  Let $0 < \sigma < 1/2$. Then
  \begin{displaymath}
  \sum_{i=0}^{\sigma k} \binom{k}{i} \leq 2^{h(\sigma) k}.
  \end{displaymath}
\end{lemma}

%\begin{IEEEproof}
%It holds that
%  \begin{align*}
%    2^{-h(\sigma) k} &= 2^{(\sigma \log \sigma + (1-\sigma) \log(1-\sigma)) k}
%    \\
%    &= \sigma^{\sigma k} (1-\sigma)^{(1-\sigma) k} \\
%    &\leq \sigma^i (1-\sigma)^{k-i} \quad \text{for $i = 0, \dots, \sigma k$.}
%\end{align*}
%where the last inequality is valid for $\sigma < 1/2$.
%
%Hence
%$$2^{-h(\sigma) k} \sum_{i=0}^{\sigma k} \binom{k}{i} \leq \sum_{i=0}^{\sigma
%        k} \binom{k}{i} \sigma^i (1 - \sigma)^{k-i}  = 1.$$
%This proves the lemma.
%\end{IEEEproof}
%
The following lemma by Rompel will be also useful.

\begin{lemma}[\cite{Rompel90}]\label{lem:rom}
Suppose $t$ is a positive even integer, $X_1, \cdots, X_u$ are $t$-wise independent random variables taking values in the range $[0,1]$, $X=\sum_{i=1}^u X_i$, $\mu = E[X]$, and $A>0$. Then
$$\Prob{|X-\mu|>A}<O\left(\left(\frac{t\mu + t^2}{A^2}\right)^{t/2}\right).$$
\end{lemma}

\section{Security Model}\label{sec:model}

\subsection{Bounded Storage Model} 

We work in a two-party scenario, where two players (Alice and Bob) engage in cryptographic protocols, more specifically commitment and oblivious transfer protocols. We assume that one of the players has an upper bound on the available memory. As usual in the cryptographic literature, we assume an adversary that can corrupt either party. We will call a corrupted party \emph{dishonest}.  Parties that have not been corrupted will be called \emph{honest}.

Cryptographic protocols in the bounded storage model run a transmission phase prior to their main part. We briefly describe this phase here.

\textbf{Transmission Phase:}
In this phase, the sender (Alice) has access to a sample $x\in \bits^n$ from an
$\alpha n$-source $X$, where $0 < \alpha < 1$, and the receiver (Bob) to $\widetilde{x} \in \bits^n$
such that $\HD(x,\widetilde{x}) \leq \delta n$. Note that this captures both 
the situation where the source is noisy and  
the situation where an adversary controls part of the source. In the bounded storage model normally a memory bound is imposed 
on both parties during this phase, but we are able to prove the security of our protocols while imposing a memory bound on only one of them 
(which one depends on the specific protocol). For a memory bounded Alice, she computes a randomized function $f(x)$ with output size at most $\gamma n$ for 
$\gamma < \alpha$, stores its output and discards $x$. Similarly, for a memory bounded Bob, he computes 
a randomized function $\widetilde{f}(\widetilde{x})$ with output size at most $\gamma n$ for 
$\gamma < \alpha$, stores its output and discards $\widetilde{x}$. We should mention that in all the proposed protocol the honest parties only have to store a bounded 
amount of information. It should also be highlighted that even if the memory bounded party
 gains infinite storage power after the 
transmission phase is over and the source is not available anymore, 
this does not affect the security of the protocol, i.e., it has everlasting security.

\subsection{Secure Commitment} 

The main part of a commitment protocol has two phases: commitment and opening.

\textbf{Commitment Phase:} Alice has an input string $v \in \calV$
(which is a realization of a random variable $V$) that
she wants to commit to. The parties exchange messages, possibly in several 
rounds. Let $\trans{\CP}{v}$ denote all the communication in this phase 
and $\viewbob{\CP}{v}$ Bob's view at the end of this phase.
These random variables are a function of $v$, the functions that the parties computed from the public random source and the parties' local randomness.

\textbf{Opening Phase:} Alice sends Bob the string $\widetilde{v}$
that she claims she committed to. The parties can then exchange 
messages in several rounds. Let $\trans{\OP}{\widetilde{v}}$ denote all the communication in this phase.
In the end Bob performs a test
$$\test\left(\viewbob{\CP}{v},\trans{\OP}{\widetilde{v}}\right)$$ that outputs 1 if Bob accepts Alice's 
commitment and 0 otherwise. 

\textbf{Security.} A commitment protocol is called $(\lcor,\lhid,\lbin)$-secure 
if it satisfies the following properties:

\begin{enumerate}
\item $\lcor$-correct: if Alice and Bob are honest, then for every possible $v$, the probability that the protocol aborts is at most $\lcor$
\ifnum\breakequations=1
	\begin{eqnarray*}
		\Prob{\text{no aborts and }\test\left(\viewbob{\CP}{v},\trans{\OP}{v}\right)= 1}\\
		 \geq 1 - \lcor.
	\end{eqnarray*}
\else
	$$\Prob{\text{no aborts and }\test\left(\viewbob{\CP}{v},\trans{\OP}{v}\right)= 1} \geq 1 - \lcor.$$
\fi

\item $\lhid$-hiding: if Alice is honest then Bob's knowledge on her committed value is at most $\lhid$, $$I(V; \viewbob{\CP}{V} | \widetilde{X}) \leq \lhid.$$

\item $\lbin$-binding: if Bob is honest, then there are no $v$ and $\widetilde{v} \neq \hat{v}$ that can be successfully open,  
$$\Prob{\test\left(\viewbob{\CP}{v},\trans{\OP}{\widetilde{v}}\right)= 1} \geq \lbin$$
and $$\Prob{\test\left(\viewbob{\CP}{v},\trans{\OP}{\hat{v}}\right)= 1} \geq \lbin.$$
\end{enumerate}

\subsection{Secure Oblivious Transfer}

We use the definition of oblivious transfer presented
in~\cite{TCC:DHRS04}. An oblivious transfer protocol is a protocol between
two players, Alice and Bob, in which Alice inputs two strings $s_0, s_1 \in
\calV$ and outputs nothing, and Bob inputs $c \in \{0,1\}$ and outputs $s
\in \{\bot,s_c\}$. Let $\viewalice{\OT}{s_0,s_1;c}$ denote the view of an Alice that interacts with an honest Bob. 
Similarly, let $\viewbob{\OT}{s_0,s_1;c}$ denote the view of a Bob that interacts with an honest Alice. 

Intuitively, the protocol will be secure for Bob if the view of Alice does not
depend on the choice bit $c$, and secure for Alice if Bob cannot obtain any
information about $s_{1-c}$. However this is tricky to formalize, because a
malicious Bob could choose to play with a different bit, depending on the public
random source and the messages exchanged before any secret is used by Alice.

In order to have more generality, the main part of the oblivious transfer protocol is divided in two phase: the setup phase, which encompass all communication 
before Alice first uses her secrets, and the transfer phase, which happens from that point on. Two pairs of inputs  $(s_0,s_1),(s_0',s_1')$ are called $i$-consistent if 
$s_i = s_i'$ for $i \in \{0,1\}$. By the end of the setup phase there should exist a random variable $I$, such that for any two $I$-consistent pairs of inputs, the resulting view of Bob is 
statistically close.\\

\noindent\textbf{Security:} A protocol is called $(\lcor,\lbob,\lalice)$-secure if it 
satisfies the following properties:

\begin{enumerate}
\item $\lcor$-correct: if Alice and Bob are honest, then
  \begin{displaymath}
    \Prob{\text{no aborts and } s = s_c} \geq 1 - \lcor
  \end{displaymath}

\item $\lbob$-secure for Bob: for any strategy used by Alice,
  \begin{displaymath} 
    \left\|\left\{\viewalice{\OT}{s_0,s_1;0}\right\} -
      \left\{\viewalice{\OT}{s_0,s_1;1} \right\}\right\|
    \leq \lbob
  \end{displaymath}

\item $\lalice$-secure for Alice: for any strategy used by Bob with
  input $c$, there exists a random variable $I$, defined at the end of the setup
  stage, such that for every two $I$-consistent pairs $(s_0,s_1), (s_0',s_1')$,
  we have
  \begin{displaymath}
    \left\|\left\{\viewbob{\OT}{s_0,s_1;c} \right\} -
    \left\{\viewbob{\OT}{s'_0,s'_1;c} \right\}\right\| \leq \lalice
  \end{displaymath}
\end{enumerate}

\section{Commitment Protocols}\label{sec:commit}

In this section we present our three commitment protocols. There are tradeoffs between the range of noise that could be tolerated and the round complexity of the proposed protocols. 
In short, the first protocol only has one message from Bob to Alice, but tolerates less noise than the second protocol, which has more rounds. Both of these protocols impose 
a memory bound on Bob, the receiver of the commitment protocol, but nothing is assumed about the memory capacity of Alice. The third protocol then deals with the complementary case in which a memory
bound is assumed on Alice, the committer, but no memory bound is assumed on Bob. Note that in many scenarios one of the parties is much more powerful than the other, so it is quite useful to
have commitments protocols working in both directions and just assume a memory bound on the weaker party.

\subsection{A Simple String Commitment Protocol}\label{sec:protocol-sc}

Next we present a quite simple string commitment protocol that only involves one message from 
Bob to Alice. A memory bound on Bob is assumed.
The scheme works as follows. First, both parties sample a number of bits from
the public source. Alice then extracts the randomness of her sample and uses it
to conceal her commitment before sending it to Bob. This guarantees the hiding
condition. She also computes a hash of her sample, where the hash function is
chosen by Bob. Alice sends Bob the concealed commitment along with the hash
value. In the open phase, Alice sends her committed value and her sampled
string. Bob then performs a number of checks for consistency. These checks
enforce binding. The details of the protocol are presented below. 

The security parameter is $\ell$ and $k$ is set as $k=2\sqrt{\ell n}$ in order to satisfy the requirements of Lemma \ref{lem:birthday}.
Fix $\varepsilon'>0$ and let $\rho = \alpha - \gamma - \frac{1 + \log(1/\varepsilon')}{n}$ according to the requirements of Lemma \ref{lem:gr}. 
To satisfy the requirements of Lemma \ref{lem:vadhan_minentropy}, fix $\tau$ such that $\frac{\rho}{3} \geq \tau > 0$, and also fix $\omega, \zeta > 0$ such that
$\rho-3\tau>\omega > 2h(\delta+\zeta)$ and $\delta+\zeta < 1/2$. Let $k_E = (\rho -3\tau - \omega) k$ and for 
$\psi >0$, $m = (1-\psi)k_E$ be the parameters of the strong extractor (Lemma \ref{lem:ext}). The message space is $\calV = \bits^m$.
It is assumed that the following functionalities, which are 
possible due to the lemmas in~\secref{sec:prelim}, are available to the parties:

\begin{itemize}
  
\item A family $\calG$ of 2-universal hash functions $g \colon \{0,1\}^k \to
\{0,1\}^{\omega k}$.

\item A $(k_E, \varepsilon_E)$-strong extractor $\Ext \colon \{0,1\}^k \times
\{0,1\}^{r} \to \{0,1\}^{m}$, for an arbitrary $\varepsilon_E>e^{-k/2^{O(\log^* k)}}$.
\end{itemize}

\begin{remark}
Note that it should hold that $2h(\delta) < \omega + 3\tau < \rho < \alpha - \gamma$, 
so the  protocol is only possible if $2h(\delta) < \alpha - \gamma$. 
\end{remark}

\subsection*{Transmission phase:}
\begin{enumerate}

\item Alice chooses uniformly $k$ positions from $X$. Similarly, Bob samples 
$k$ positions from $\widetilde{X}$. We call their sets of positions $\calA$ and $\calB$, respectively.
\end{enumerate}

\subsection*{Commit phase:}
\begin{enumerate}

\item Alice announces $\calA$ to Bob.

\item Bob chooses $g \getsr \calG$ and sends its description to Alice.

\item Alice computes $p \gets g(x^\calA)$, $u \getsr \{0,1\}^{r}$, and $y \gets
\Ext(x^\calA,u)$. She then computes $z = v \oplus y$ and sends $(z, 
p, u)$ to Bob in order to commit to $v$.
\end{enumerate}

\subsection*{Open phase:}

\begin{enumerate}
\item Alice sends $v'$ and $w$ to Bob, which are defined as $v'=v$ and 
$w=x^\calA$ in the case that she is honest.

\item Let $\calC = \calA \cap \calB$, $c = |\calC|$ and $w^\calC$ be the restriction of $w$
to the positions corresponding to the set $\calC$. Bob verifies whether 
$c \geq \ell$, $\HD(w^\calC,\widetilde{x}^\calC) \leq (\delta + \zeta) c$,
$p=g(w)$ and $v' = \Ext(w,u) \oplus z$. If any 
verification fails Bob outputs $0$, otherwise he outputs $1$.
\end{enumerate}

\begin{theorem}\label{th:sc}
The protocol is $(\lcor,\lhid,\lbin)$-secure for $\lcor,\lhid$ and $\lbin$ negligible in $\ell$.
\end{theorem}

\begin{IEEEproof} \textbf{Correctness:}
It is clear that if both Alice and Bob are honest, the protocol will fail only
in the case that $c<\ell$ or $\HD(x^\calC,\widetilde{x}^\calC) > (\delta + \zeta) c$. 
By \lemref{lem:birthday}, $c\geq \ell$ except with probability at most $e^{-\ell/4}$.
By \lemref{lem:alice-bob-random-subset}, $\HD(x^\calC,\widetilde{x}^\calC) \leq
(\delta+\zeta) c$ except with probability at most $e^{-c \zeta^2/2}$, which is negligible
in $\ell$ if $c\geq \ell$. 

\noindent\textbf{Hiding:} After the commit phase, (a possibly malicious) Bob possesses 
$(z, p, \calA, u)$, $g$ and the output of a function $\widetilde{f}(\cdot)$
of $\widetilde{x}$, where $|\widetilde{f}(\widetilde{x})| \leq \gamma n$ with 
$\gamma < \alpha$. The only random variable that can provide mutual
information about $V$ when conditioned on $\widetilde{X}$ is $Z$, but we prove below 
that $Z$ is almost uniform from Bob's point of view, and so 
it works as an one-time pad and only negligible information can be leaked.

By \lemref{lem:gr},
\begin{displaymath}
  R_{\infty}^{\varepsilon'}(X | \widetilde{f}(\widetilde{X})) \geq \alpha - \gamma -
  \frac{1+\log(1/\varepsilon')}{n} = \rho.
\end{displaymath}

Since Alice chooses $\calA$ randomly and this is an $(\mu,\nu,e^{-k \nu^2/2})$-averaging 
sampler for any $\mu,\nu >0$ according to \lemref{lem:random-subset}. By setting 
$\mu=\frac{\rho-2\tau}{\log (1/ \tau)}$, $\nu=\frac{\tau}{\log (1/ \tau)}$, we have 
by \lemref{lem:vadhan_minentropy} that 
\begin{displaymath}
R_{\infty}^{\varepsilon''+\varepsilon'}(X^\calA | \calA, \widetilde{f}(\widetilde{X})) \geq \rho - 3\tau,
\end{displaymath}
where $\varepsilon''$ is a negligible function of $k$.

It holds that
%%%%%%%%%%%%%%%%%%%%%%%%%%%%%%%
% If changing, change both
%%%%%%%%%%%%%%%%%%%%%%%%%%%%%%%
\ifnum\breakequations=1
	\begin{eqnarray*}
	 \! &&H_{\infty}^{\varepsilon''+\varepsilon'}(X^\calA | G(X^\calA), \calA, U, G,  \widetilde{f}(\widetilde{X})) \\
	 &  &  \qquad \qquad = H_{\infty}^{\varepsilon''+\varepsilon'}(X^\calA | G(X^\calA), \calA, \widetilde{f}(\widetilde{X}))\\
	 & &  \qquad \qquad \geq  H_{\infty}^{\varepsilon''+\varepsilon'}(X^\calA | \calA, \widetilde{f}(\widetilde{X})) - H_0(G(X^\calA))\\
	 &&   \qquad \qquad \geq (\rho - 3\tau - \omega) k\\
	 & &  \qquad \qquad = k_E.
	\end{eqnarray*}
\else
	\begin{eqnarray*}
	H_{\infty}^{\varepsilon''+\varepsilon'}(X^\calA | G(X^\calA), \calA, U, G, \widetilde{f}(\widetilde{X})) &=&
	H_{\infty}^{\varepsilon''+\varepsilon'}(X^\calA | G(X^\calA), \calA, \widetilde{f}(\widetilde{X}))\\
	& \geq & H_{\infty}^{\varepsilon''+\varepsilon'}(X^\calA | \calA, \widetilde{f}(\widetilde{X})) - H_0(G(X^\calA))\\
	& \geq & (\rho - 3\tau - \omega) k\\
	&=& k_E.
	\end{eqnarray*}
\fi

Therefore, setting $\varepsilon'$ and $\varepsilon_E$
to be negligible in $\ell$, the use of the strong extractor to obtain $y$
(and of $y$ to xor the message) guarantees that only negligible information 
about the committed message can be leaked.

\textbf{Binding:}
The protocol is binding if, after the commit phase, Alice cannot 
choose between two different values to successfully open. 
Let $\sigma = \delta +\zeta$. 
The only way Alice can cheat is if she can come up with two strings 
$w, w'$ such that $g(w) = g(w')$, $\HD(w^\calC,\widetilde{x}^\calC) \leq 
\sigma c$ and $\HD(w'^\calC,\widetilde{x}^\calC) \leq \sigma c$ (with $c\geq \ell$).
If this happens, it holds that either there are two strings 
$w, w'$ such that $g(w) = g(w')$, $\HD(w,\widetilde{x}^\calA) \leq 
\sigma k$ and $\HD(w',\widetilde{x}^\calA) \leq \sigma k$; or Alice can 
compute $w$ (without knowing the set $\calB$ that together
with $\calA$ determines $\calC$) such that $\HD(w,\widetilde{x}^\calA) > \sigma k$ 
and $\HD(w^\calC,\widetilde{x}^\calC) \leq \sigma c$.
It is proven below that the probability that Alice succeeds in cheating 
decreases exponentially with the security parameter $\ell$ (or, 
equivalently in $k,c$). First the probability that there exists two different 
strings $w, w'$ both within Hamming distance $\sigma k$ from 
$\widetilde{x}^\calA$ and such that $g(w) = g(w')$ is upper bounded by
%%%%%%%%%%%%%%%%%%%%%%%%%%%%%%%
% If changing, change both
%%%%%%%%%%%%%%%%%%%%%%%%%%%%%%%
\ifnum\breakequations=1
	\begin{eqnarray*}
 	&& \Prob{\exists w, w'~\suchthat
	\left\{
	\begin{array}{l}
	w \neq w'\\  
	g(w) = g(w')\\
	\HD(w,\widetilde{x}^\calA) \leq \sigma k\\
	\HD(w',\widetilde{x}^\calA) \leq \sigma k
	\end{array}
	\right.
	} = \\
	&  &\qquad =\sum_{w \suchthatt \HD(w,\widetilde{x}^\calA) \leq \sigma k}
	\left(\sum_{w' \neq w \suchthatt \HD(w',\widetilde{x}^\calA) \leq \sigma k} 2^{-\omega k}
	\right)\\	
	&& \qquad  \leq  2^{-(\omega-2 h(\sigma))k}
	\end{eqnarray*}
\else
	\begin{eqnarray*}
	\Prob{\exists w, w'~\suchthat
	\left\{
	\begin{array}{l}
	w \neq w'\\  
	g(w) = g(w')\\
	\HD(w,\widetilde{x}^\calA) \leq \sigma k\\
	\HD(w',\widetilde{x}^\calA) \leq \sigma k
	\end{array}
	\right.
	}
	&=& \sum_{w \suchthatt \HD(w,\widetilde{x}^\calA) \leq \sigma k}
	\left(\sum_{w' \neq w \suchthatt \HD(w',\widetilde{x}^\calA) \leq \sigma k} 2^{-\omega k}
	\right)\\
	& \leq & 2^{-(\omega-2 h(\sigma))k},
	\end{eqnarray*}
\fi

where \lemref{lem:bound_binom} was used to obtain the inequality.
By design, it holds that $\omega > 2 h(\sigma)$, therefore the probability 
that Alice successfully cheats by finding two strings that are at distance 
at most $\sigma k$ from $\widetilde{x}^\calA$ and hash to the same value is 
negligible in $k$.

Now considering the second case, by assumption $w$ has Hamming
distance $(\sigma +\psi) k$ from $\widetilde{x}^\calA$ for some $\psi > 0$. Since Bob 
is honest, $\calB$ is chosen randomly. Hence \lemref{lem:alice-bob-random-subset} 
can be applied and thus the probability that
$\HD(w^\calC,\widetilde{x}^\calC) \leq \sigma c$ is smaller than $e^{-c\psi^2/2}$.
\end{IEEEproof}

%\begin{remark}
%it is possible to relax the condition $2h(\delta) < \alpha - \gamma$ to $h(\delta) < \alpha - \gamma$ in the case when Alice behave honestly in the commitment phase, but is free tocheat in the remaining of the protocol. This follows from 
%the fact that, in order to break the binding condition, the adversary has to find a string that has small Hamming distance and hashes to one \emph{specific} 
%value, instead of finding any two strings with small Hamming distance that hash to the same value.
%\end{remark}

\subsection{Extending the Feasibility Region}\label{sec:extend}

While the previous protocol is simple, efficient and round optimal, it works for a rather limited range of noise: $h(\delta) < (\alpha - \gamma)/2$. We next present a more elaborate protocol that works for a much larger range of noise $h(\delta) < \alpha - \gamma$ at the cost of increasing the rounds of communication. The memory bound is still on Bob. The idea for guaranteeing the binding property is to use two rounds of hash challenge-responses in order to guarantee the binding condition. Consider the initial set of viable strings that Alice can possibly send to Bob during the commitment phase that would pass the Hamming distance test. The first hash challenge-response round binds Alice to one specific output of the hash function, and thus restrict the set of viable strings to be polynomial in the security parameter. The second hash challenge-response round then binds Alice to one specific value for the commitment. Our solution is based on families of $4k$-universal hash functions. 
 This approach has been used before in a different context \cite{EC:DamKilSal99}.

The security parameter is $\ell$ and $k$ is set as $k=2\sqrt{\ell n}$ in order to satisfy the requirements of Lemma \ref{lem:birthday}.
Fix $\varepsilon'>0$ and let $\rho = \alpha - \gamma - \frac{1 + \log(1/\varepsilon')}{n}$ according to the requirements of Lemma \ref{lem:gr}. 
To satisfy the requirements of Lemma \ref{lem:vadhan_minentropy}, fix $\tau$ such that $\frac{\rho}{3} \geq \tau > 0$, and also fix $\omega_1,\omega_2, \zeta > 0$ such that $\rho-3\tau>\omega_1+\omega_2$,
$\omega_1 > h(\delta+\zeta)$,  and $\delta+\zeta < 1/2$. Let $k_E = (\rho -3\tau - \omega_1 - \omega_2) k$ and for 
$\psi >0$, $m = (1-\psi)k_E$ be the parameters of the strong extractor (Lemma \ref{lem:ext}). The message space is $\calV = \bits^m$.
It is assumed that the following functionalities, which are 
possible due to the lemmas in~\secref{sec:prelim}, are available to the parties:

\begin{itemize}
  
\item A family $\calG_1$ of $4k$-universal hash functions $g_1 \colon \{0,1\}^k \to
\{0,1\}^{\omega_1 k}$.

\item A family $\calG_2$ of $2$-universal hash functions $g_2 \colon \{0,1\}^k \to
\{0,1\}^{\omega_2 k}$.

\item A $(k_E, \varepsilon_E)$-strong extractor $\Ext \colon \{0,1\}^k \times
\{0,1\}^{r} \to \{0,1\}^{m}$, for an arbitrary $\varepsilon_E>e^{-k/2^{O(\log^* k)}}$.
\end{itemize}

\begin{remark}
Note that it should hold that $h(\delta) < \omega_1 + 3\tau < \rho < \alpha - \gamma$, 
so the  protocol is only possible if $h(\delta) < \alpha - \gamma$. 
\end{remark}

\subsection*{Transmission phase:}
\begin{enumerate}

\item Alice chooses uniformly $k$ positions from $X$. Similarly, Bob samples 
$k$ positions from $\widetilde{X}$. We call their sets of positions $\calA$ and $\calB$, respectively.
\end{enumerate}

\subsection*{Commit phase:}
\begin{enumerate}

\item Alice announces $\calA$ to Bob.

\item Bob chooses $g_1 \getsr \calG_1$ and sends its description to Alice.

\item Alice computes $p_1 \gets g_1(x^\calA)$ and sends it to Bob.

\item Bob chooses $g_2 \getsr \calG_2$ and sends its description to Alice.

\item Alice computes $p_2 \gets g_2(x^\calA)$, $u \getsr \{0,1\}^{r}$, and $y \gets
\Ext(x^\calA,u)$. She then computes $z = v \oplus y$ and sends $(z, 
p_2, u)$ to Bob in order to commit to $v$.
\end{enumerate}

\subsection*{Open phase:}

\begin{enumerate}
\item Alice sends $v'$ and $w$ to Bob, which are defined as $v'=v$ and 
$w=x^\calA$ in the case that she is honest.

\item Let $\calC = \calA \cap \calB$, $c = |\calC|$ and $w^\calC$ be the restriction of $w$
to the positions corresponding to the set $\calC$. Bob verifies whether 
$c \geq \ell$, $\HD(w^\calC,\widetilde{x}^\calC) \leq (\delta + \zeta) c$,
$p_1=g_1(w)$, $p_2=g_2(w)$ and $v' = \Ext(w,u) \oplus z$. If any 
verification fails Bob outputs $0$, otherwise he outputs $1$.
\end{enumerate}

\begin{theorem}
The protocol is $(\lcor,\lhid,\lbin)$-secure for $\lcor,\lhid$ and $\lbin$ negligible in $\ell$.
\end{theorem}

\begin{IEEEproof} \textbf{Correctness:} Same as in \thref{th:sc}.

\noindent\textbf{Hiding:} Follows the same lines as in \thref{th:sc}. The difference is that here $k_E = (\rho -3\tau - \omega_1 - \omega_2) k$ 
in order to account for the entropy loss due to the output of both hash functions $g_1$ and $g_2$ (instead of $k_E = (\rho -3\tau - \omega)$ in \thref{th:sc} that accounts for the output of a single hash function $g$).

\textbf{Binding:}
The protocol is binding if, after the commit phase, Alice cannot 
choose between two different values to successfully open. 
Let $\sigma = \delta +\zeta$. 
The only way Alice can cheat is if she can come up with two different strings 
$w, w'$ that pass all tests performed by Bob during the opening phase. Either $\HD(w,\widetilde{x}^\calA) \leq 
\sigma k$ and $\HD(w',\widetilde{x}^\calA) \leq \sigma k$; or Alice can 
compute $w$ (without knowing the set $\calB$ that together
with $\calA$ determines $\calC$) such that $\HD(w,\widetilde{x}^\calA) > \sigma k$ 
and $\HD(w^\calC,\widetilde{x}^\calC) \leq \sigma c$. The probability that Alice succeeds in cheating
in the latter case can be upper bounded as in \thref{th:sc}. Below we upper bound her cheating success probability in the former case and prove that it 
decreases exponentially with the security parameter $\ell$ (or, 
equivalently in $k$). 

Let the viable set dynamically denote the strings that Alice can possibly send to Bob with non-negligible probability of successful opening. Before the first round of hash challenge-response, the viable set consists of all $w$ such that $\HD(w,\widetilde{x}^\calA) \leq \sigma k$. Now lets consider an arbitrary fixed value $p_1$ for the output of the first hash. Considering the $j$-th viable string before the first hash challenge-response round, define $I_j$ as $1$ if the $j$-th viable string is mapped by $g_1$ to $p_1$; otherwise $I_j=0$. And define $I =\sum_j I_j$. Clearly $\mu=E[I]<1$, as $g_1$ is chosen from a $4k$-universal family of hash functions with range of size $\{0,1\}^{\omega_1 k}$ for $\omega_1 > h(\delta+\zeta)$. Let $p_1$ be called bad if $I$ is bigger than $8k+1$. Using the fact that $g_1$ is $4k$-wise independent and applying \lemref{lem:rom} with $t=4k$ and $A=2t=8k$, we get
\begin{eqnarray*}
\Prob{I>8k+1}&<&O\left(\left(\frac{t\mu + t^2}{(2t)^2}\right)^{t/2}\right)\\
&<&O\left(\left(\frac{1 + t}{4t}\right)^{t/2}\right)\\
&<&O\left(2^{-t/2}\right).
\end{eqnarray*}

Then the probability that any $p_1$ is bad is upper bounded by 
$$O\left(2^{\omega_1 k}2^{-t/2}\right) < O\left(2^{-k}\right).$$

If the viable set is reduced to at most $8k+1$ elements after the first hash challenge-response round, then the probability that some of those collide 
in the second hash challenge-response round is upper bounded by
$$\left(8k+1\right)^22^{-\omega_2 k},$$
which is negligible in $k$.
\end{IEEEproof}

\subsection{Alternative Bit Commitment Protocol}\label{sec:protocol-bc}

Next we design a $\emph{bit}$ commitment protocol where the memory bound is imposed on Alice instead of Bob. The protocol works for $h(\delta) < \alpha - \gamma$ and uses (the cheaper) families of 2-universal hash functions, instead of $4k$-universal hash functions. The central idea is to use an interactive hashing execution to perform the bit commitment~\cite{IMA:ShiYam11}.

Before describing our solution, we remark that is important to obtain protocols that work for memory bounded Alice and protocols that work for memory bounded Bob. This is particularly interesting in the case of an asymmetry of power between the parties - when one of the parties is much more powerful than the other. It makes sense to impose the bound on the weak party, whenever it is the sender of the commitment (Alice) or the receiver of the commitment (Bob).

Alice has a bit $v$ which she wants to commit to.
The security parameter is $\ell$ and $k$ is set as $k=2\sqrt{\ell n}$ in order to satisfy the requirements of Lemma \ref{lem:birthday}. Fix $\varepsilon'>0$ and $\xi >0$ such that $\delta + \xi < 1/2$, and let 
$\rho = \alpha - \gamma - \frac{1 + \log(1/\varepsilon')}{n}$ according to the requirements of Lemma \ref{lem:gr}. Fix $0<\zeta<1$
and $\tau$ such that $\frac{\rho}{3} \geq \tau > 0$ to satisfy the requirements of Lemma \ref{lem:vadhan_minentropy}. Let $\mu=\frac{\rho-2\tau}{\log (1/ \tau)}$, 
$\nu=\frac{\tau}{\log (1/ \tau)}$ and $\varepsilon''=e^{-\ell \nu^2/2}-2^{-\Omega(\tau n)}$, where the last term comes 
from \lemref{lem:vadhan_minentropy}. Fix  $m \geq \ell \left(\log{k} + 1\right)$ and $m - O(\ell) \geq t \geq m - \zeta \log (1/(\varepsilon'+\varepsilon''))$ according to Lemma \ref{lem:interactive-hashing-ding}. 
It is assumed that the following functionality, which is
possible due to the lemmas in~\secref{sec:ih}, is available to the parties:

\begin{itemize}  
\item An $2^{-m}$-uniform $(t,2^{-(m-t)+O(\log m)})$-secure interactive hashing protocol with input domain $\calW = \bits^m$ and an associated dense encoding of subsets $F$
for tuples of size $k$ and subsets of size $\ell$.
\end{itemize}

The following bit commitment protocol is correct and secure if $h(\delta+\xi) < \rho - 3\tau$.

\subsection*{Transmission phase:}
\begin{enumerate}

\item Alice chooses uniformly $k$ positions from $X$. Similarly, Bob samples 
$k$ positions from $\widetilde{X}$. We call their sets of positions $\calA$ and $\calB$, respectively.
\end{enumerate}

\subsection*{Commit phase:}
\begin{enumerate}

\item Bob announces $\calB$ to Alice. Alice computes $\calD = \calA \cap \calB$. If 
$|\calD| < \ell$, Alice aborts. Otherwise, Alice picks a random 
subset $\calC$ of $\calD$ of size $\ell$.

\item Alice computes the encoding $w$ of $\calC$ (as a subset of $\calB$). Alice and 
Bob interactively hash $w$, producing two strings $w_0,w_1$. 
They compute the subsets $\calC_0, \calC_1 \subset \calB$ that are respectively encoded in 
$w_0, w_1$. If either encoding is invalid, they abort.

\item Alice sends $p = v \oplus d$ to Bob, where $w_d = w$. 
\end{enumerate}

\subsection*{Open phase:}

\begin{enumerate}
\item Alice sends $v'$ and $x'^{\calC'}$ to Bob, which are defined as $v'=v$ and 
$x'^{\calC'}=x^{\calC}$ in the case that she is honest.

\item Bob computes $d' = p \oplus v'$ and checks whether $\HD(x'^{\calC'},\widetilde{x}^{\calC_{d'}}) \leq (\delta + \xi) \ell$. 
If the verification fails Bob outputs $0$, otherwise he outputs $1$.
\end{enumerate}

\begin{theorem}
The protocol is $(\lcor,0,\lbin)$-secure for $\lcor$ and $\lbin$ negligible in $\ell$.
\end{theorem}

\begin{IEEEproof} \textbf{Correctness:}
If both participants are honest, the protocol fails only in the following cases: 
(1) $|\calD|<\ell$; (2) $\HD(x^\calC,\widetilde{x}^\calC) > (\delta + \xi) \ell$ or (3) $w_0$ or $w_1$ is 
an invalid encoding of a subset. 
By \lemref{lem:birthday}, $|\calD|\geq \ell$ except with probability at most $e^{-\ell/4}$.
By \lemref{lem:alice-bob-random-subset}, $\HD(x^\calC,\widetilde{x}^\calC) \leq
(\delta+\xi) \ell$ except with probability at most $e^{-\ell \xi^2/2}$.
Finally, since $w_d = w$ is the encoding of $C$, one of the two outputs of the interactive 
hashing protocol is always a valid encoding. The other output $W_{1-d}$ is $2^{-m}$-close 
to distributed uniformly over the $2^{-m} - 1$ strings different 
from $w_d$. Since it is a dense encoding, \lemref{lem:dense-encoding} implies
that the probability that it is not a valid encoding is thus less than or equal to
\ifnum\breakequations=1
	\begin{eqnarray*}
		2^{-m} + \frac{\binom{k}{\ell}}{ 2^{m}-1} &\leq& 2^{-m} + 2^{\ell \log k - m+1} \\
		&\leq& 2^{-\ell \log k - \ell}+ 2^{-\ell+1}\\ 
		&\leq&  2^{-\ell+2},
	\end{eqnarray*}
\else
	$$2^{-m} + \frac{\binom{k}{\ell}}{ 2^{m}-1} \leq 2^{-m} + 2^{\ell \log k - m+1} \leq 2^{-\ell \log k - \ell}
	+ 2^{-\ell+1} \leq  2^{-\ell+2}$$,
\fi
for $m \geq \ell \left(\log k+1\right)$.

Putting everything together this proves the correctness.

\noindent\textbf{Hiding:}
There are two possibilities: either the protocol does not abort; or it  aborts due to $|\calD|<\ell$ or an invalid encoding. If the protocol aborts, Alice still has not sent 
$p = v \oplus d$, so Bob's view is independent from $V$. 
On the other hand, if the protocol does not abort, then $w_{1-d}$ is a valid encoding
of some set $\calC'$. Due to the properties of the interactive hashing
protocol, Bob's view is then consistent with both
\begin{enumerate}
\item Alice committing to $v$ and $\calC$ being the subset for which she knows the positions of $x$, and
\item Alice committing to $1-v$ and $\calC'$ being the subset for which she knows the positions of $x$.
\end{enumerate}
Hence Bob's view is independent of $V$. 

\noindent\textbf{Binding:}
The strategy of the proof is to demonstrate that there is an $i$ such $X^{\calC_{i}}$ has 
high enough min-entropy from Alice's point of view so that she cannot guess (except with negligible 
probability) a string $X'^{\calC_{i}}$ that is close enough to $\widetilde{X}^{\calC_{i}}$. Hence she will not be 
able to successfully use this output of the interactive hashing during the opening phase and will thus be bounded 
to use the other output of the interactive hashing. 
By the bounded storage assumption, the bounded information $f(X)$ stored by Alice is such that $|f(X)| \leq \gamma n$  with
$\gamma < \alpha$. Then, by \lemref{lem:gr},

\begin{displaymath}
  R_{\infty}^{\varepsilon'}(X | f(X)) \geq \alpha - \gamma -
  \frac{1+\log(1/\varepsilon')}{n} = \rho.
\end{displaymath}

Since Bob is honest, $\calB$ is randomly chosen. Lets consider a random 
subset $\widetilde{\calC}$ of $\calB$ such that  $|\widetilde{\calC}|=\ell$. This is an 
$(\mu,\nu,e^{-\ell \nu^2/2})$-averaging sampler for any $\mu,\nu >0$ according 
to \lemref{lem:random-subset}. By setting $\mu=\frac{\rho-2\tau}{\log (1/ \tau)}$, 
$\nu=\frac{\tau}{\log (1/ \tau)}$, we have by \lemref{lem:vadhan_minentropy} that 
\begin{displaymath}
R_{\infty}^{\varepsilon'+\varepsilon''}(X^{\widetilde{\calC}} | \calB, \widetilde{\calC}, f(X)) \geq \rho - 3\tau,
\end{displaymath}
for $\varepsilon''=e^{-\ell \nu^2/2}-2^{-\Omega(\tau n)}$. For 
$\widetilde{\varepsilon}=(\varepsilon'+\varepsilon'')^{1-\zeta}$, let 
$\Bad$ be the set of $\widetilde{\calC}$'s such that $R_{\infty}(X^{\widetilde{\calC}} | \calB, \widetilde{\calC}, f(X))$
is not $\widetilde{\varepsilon}$-close to $(\rho - 3\tau)$-min entropy rate. Due to the above 
equation the density of $\Bad$ is at most $(\varepsilon'+\varepsilon'')^{\zeta}$. Then
the size of the set $T \subset \bits^m$ of strings that maps (using the dense encoding scheme) to subsets in $\Bad$ is at most 
$(\varepsilon'+\varepsilon'')^{\zeta}2^m \leq 2^t$. Hence the properties of the interactive 
hashing protocol guarantee that with overwhelming probability there
will be an $i$ such that
$$R_{\infty}^{\widetilde{\varepsilon}}(X^{\calC_{i}} | \calB, \calC_{i}, f(X), M_{IH}) \geq \rho - 3\tau,$$
where $M_{IH}$ are the messages exchanged during the interactive hashing protocol.

However, if $h(\delta+\xi) < \rho - 3\tau$ and the min-entropy rate is at least $\rho - 3\tau$, then fixing 
$0<\hat{\varepsilon}<\rho - 3\tau-h(\delta+\xi)$,
for large enough $\ell$, the probability that Alice guesses one of the strings $X'^{\calC_{i}}$ that would be accepted by Bob 
as being close enough to $\widetilde{X}^{\calC_{i}}$ is upper bounded by 
$$2^{\left(h(\delta+\xi)-\rho + 3\tau - \hat{\varepsilon}\right)\ell},$$
which is a negligible function of $\ell$.
\end{IEEEproof}

By fixing the parameters as small as possible we have that for large enough $\ell$ the protocol works for values 
$\alpha, \gamma, \delta$ which satisfy $h(\delta) < \alpha - \gamma$.

\section{Oblivious Transfer Protocol}\label{sec:protocol-ot}
 
Our OT protocol imposes a memory bound on Bob. We would like to point out that it is trivial to revert the direction of OT protocols \cite{EC:WolWul06}. 
We first present the intuition behind our protocol before a detailed description. Initially, both parties sample positions from 
the public random source. Then the parties use an interactive hashing protocol (with an associated 
dense encoding) to select two subsets of the positions initially sampled by Alice.
Bob inputs into the interactive hashing protocol one subset for which he has also
sampled the public random source in the same positions. The other subset is out of 
Bob's control due to the properties of the interactive hashing protocol. Finally the positions specified by 
the two subsets are used as input to a fuzzy extractor in order to obtain one-time pads.
Bob sends one bit indicating which input string should be xored with which
one-time pad. The security for Alice is guaranteed by the fact that one of the subsets is
out of Bob's control and will have high min-entropy given his view, thus resulting in a 
good one-time pad. The security for Bob follows from the security of the interactive hashing.
The correctness follows from the correctness of the fuzzy extractor.

The security parameter is $\ell$ and $k$ is set as $k=2\sqrt{\ell n}$  in order to satisfy the requirements of Lemma \ref{lem:birthday}. Fix $\varepsilon',\hat{\varepsilon}>0$ and 
$ \xi>0$ such that $1/4>\delta+\xi>0$
 and 
let $\rho = \alpha - \gamma - \frac{1 + \log(1/\varepsilon')}{n}$ according to the requirements of Lemma \ref{lem:gr}. 
Fix $0<\zeta<1$
and $\tau$  such that $\frac{\rho}{3} \geq \tau > 0$ to satisfy the requirements of Lemma \ref{lem:vadhan_minentropy}. Let $\mu=\frac{\rho-2\tau}{\log (1/ \tau)}$, 
$\nu=\frac{\tau}{\log (1/ \tau)}$ and $\varepsilon''=e^{-\ell \nu^2/2}-2^{-\Omega(\tau n)}$, where the last term comes 
from \lemref{lem:vadhan_minentropy}. Fix $m \geq \ell \left(\log{k} + 1\right)$ and $m - O(\ell) \geq t \geq m - \zeta \log (1/(\varepsilon'+\varepsilon''))$ according to Lemma \ref{lem:interactive-hashing-ding}. 
For $\beta$ depending on $\delta+\xi$ (see comments about the code rate below), let $k_F$ and $m_F$ (the parameters of the fuzzy extractor) be such that $k_F = \rho + \beta - 3\tau-2 m_F - 1 - \frac{1+\log(1/\hat{\varepsilon})}{\ell}$ and $0 < m_F < k_F$. 
The message is $\calV=\bits^{m_F \ell}$.
We assume the following functionalities are available to the parties (see the lemmas in Sections \ref{sec:ext} and \ref{sec:ih}), :

\begin{itemize}
  
\item A pair of functions $\Ext \colon \{0,1\}^\ell \times \{0,1\}^{r} \to 
\{0,1\}^{m_F \ell} \times \{0,1\}^q$ and $\Rec \colon \{0,1\}^\ell \times \{0,1\}^{r} 
\times \{0,1\}^q \to \{0,1\}^{m_F \ell}$ that constitutes an 
$(k_F \ell,\varepsilon_F,\delta+\xi,0)$-fuzzy extractor where
 $q = (1-R) \ell$, $\varepsilon_F$ is an arbitrary number with
$\varepsilon_F>e^{-\ell/2^{O(\log^* \ell)}}$.

\item An $2^{-m}$-uniform $(t,2^{-(m-t)+O(\log m)})$-secure interactive hashing protocol with input domain $\calW = \bits^m$ and an associated dense encoding of subsets $F$
for tuples of size $k$ and subsets of size $\ell$.
\end{itemize}

Recall (\remref{rem:fuzzy-extractor}) that there is a tradeoff between the
fraction of errors $\delta+\xi$ that the fuzzy extractor can tolerate and the
rate $\beta$ of the code used in the construction. The construction given in Theorem
4 of~\cite{guruswami2002near} has linear-time encoding and decoding and achieves
the Zyablov bound \cite{zyablov71}: for given $1 > \beta > 0$ and $\mu > 0$, the code has rate $\beta$ and
  \begin{equation}
    \label{eq:zyablov-distance}
    \delta+\xi \geq \max_{\beta < \widetilde{\beta} < 1} \frac{(1 - \widetilde{\beta} - \mu) y}{2},
  \end{equation}
  where $y$ is the unique number in $[0,1/2]$ with $h(y) = 1-\beta/\widetilde{\beta}$ and $\delta+\xi$ the amount of errors that can be corrected by the code.

Note that in order for $k_F$ to be positive, we need to have $\rho + \beta> 1$;
since $\rho$ approaches $\alpha - \gamma$ from below in the asymptotic limit, we can obtain an
upper bound for $\delta$ by setting $\beta > 1-\alpha+\gamma$ and
$\mu = 0$ in~\eqref{eq:zyablov-distance}.

 Random linear codes achieve a better
bound, namely, the Gilbert-Varshamov bound: for a given relative distance
$\upsilon$ and $\mu > 0$, a random code has (with high probability) rate $\beta \geq 1 - h(\upsilon) -
\mu$. Applying again the constraint that $\rho + \beta > 1$ and that $\rho \to
\alpha - \gamma$ in the asymptotic limit, and using the fact that a code that
can correct $\delta n$ errors has relative distance $\upsilon = 2 \delta + 1/n \to 2
\delta$, this gives an upper bound for $\delta$: we must have
$h(2 \delta) < \alpha - \gamma$. However, as noted
in~\remref{rem:fuzzy-extractor}, the random linear code construction does not
have efficient decoding. It is an open question whether an
efficient construction can achieve better parameters than the one
from~\cite{guruswami2002near}.

\subsection*{Transmission phase:}
\begin{itemize}

\item Alice chooses uniformly $k$ positions from $X$. Similarly, 
Bob samples $k$ positions from $\widetilde{X}$. We call their
sets of positions $\calA$ and $\calB$, respectively.
\end{itemize}

\subsection*{Setup phase:}
\begin{itemize}
\item Alice sends $\calA$ to Bob. Bob computes $\calD = \calA \cap \calB$. If 
$|\calD| < \ell$, Bob aborts. Otherwise, Bob picks a random 
subset $\calC$ of $\calD$ of size $\ell$.

\item Bob computes the encoding $w$ of $\calC$ (as a subset of $\calA$). Alice and 
Bob interactively hash $w$, producing two strings $w_0,w_1$. 
They compute the subsets $\calC_0, \calC_1 \subset A$ that are respectively encoded in 
$w_0, w_1$. If either encoding is invalid, they abort.
\end{itemize}

\subsection*{Transfer phase:}
\begin{itemize}

\item Bob sends $p = c \oplus d$, where $w_d = w$. 
\item For $i \in \{0,1\}$, Alice picks $r_i \getsr \bits^{r}$,
computes $(y_i, q_i) \gets \Ext(x^{\calC_i},r_i)$ and 
$z_i = s_{i \oplus p} \oplus y_i$, and sends $(z_i,r_i,q_i)$
to Bob.
\item Bob computes $y' \gets \Rec(\widetilde{x}^{\calC}, r_{d},
q_{d})$ and outputs $s=y' \oplus z_{d}$.
\end{itemize}

\begin{theorem}
The protocol is $(\lcor,0,\lalice)$-secure for $\lcor$ and $\lalice$ negligible in $\ell$.
\end{theorem}

\begin{IEEEproof}\textbf{Correctness:}
We first analyze the probability of an abort. The protocol aborts if 
either $|\calD| < \ell$, or if one string obtained in the interactive hashing 
protocol is an invalid encoding of subsets of $\calA$. By \lemref{lem:birthday}, 
$\Pr[|\calD| < \ell] < e^{-\ell/4}$. Since $w_d = w$ is the encoding of $\calC$,
one of the two string is always a valid encoding. The other output $W_{1-d}$ is $2^{-m}$-close 
to distributed uniformly over the $2^{-m} - 1$ strings different 
from $w_d$. Since it is a dense encoding, \lemref{lem:dense-encoding} implies
that the probability that it is not a valid encoding is thus less than or equal to
\ifnum\breakequations=1
	\begin{eqnarray*}
		2^{-m} + \frac{\binom{k}{\ell}}{ 2^{m}-1} &\leq& 2^{-m} + 2^{\ell \log k - m+1}\\
		& \leq & 2^{-\ell \log k - \ell} + 2^{-\ell+1}\\
		& \leq & 2^{-\ell+2},
	\end{eqnarray*}
\else
	$$2^{-m} + \frac{\binom{k}{\ell}}{ 2^{m}-1} \leq 2^{-m} + 2^{\ell \log k - m+1} \leq 2^{-\ell \log k - \ell} + 2^{-\ell+1} \leq  2^{-\ell+2},$$ 
\fi
for $m \geq \ell \left(\log k+1\right)$.
If both parties are honest and there is no abort, then $s = s_c$ if and 
only if $\Rec(\widetilde{x}^{\calC}, r_{d},q_{d}) = y_d$. By the properties 
of the employed fuzzy extractor, this last event happens if $\HD(x^\calC,\widetilde{x}^\calC) \leq (\delta + \xi) \ell$. 
By \lemref{lem:alice-bob-random-subset}, $\HD(x^\calC,\widetilde{x}^\calC) > (\delta + \xi) \ell$
with probability at most $e^{-\xi^2 \ell/2}$. Putting everything together
this proves the correctness.

\noindent\textbf{Security for Bob:}
There are two possibilities: either the protocol aborts or not. 
If the protocol aborts in the setup phase, Bob still has not sent 
$p = c \oplus d$, so Alice's view is independent from $C$. On the
other hand, if the protocol does not abort, then $w_{1-d}$ is a valid encoding
of some set $\calC'$. Due to the properties of the interactive hashing
protocol, Alice's view is then consistent with both

\begin{enumerate}
\item Bob choosing $c$ and $\calC$, and
\item Bob choosing $1-c$ and $\calC'$.
\end{enumerate}

Hence Alice's view is independent of $C$. 

\noindent\textbf{Security for Alice:}
There should be an index $i$ (determined at the setup stage) such
that for any two pairs $(s_0,s_1), (s_0',s_1')$ with $s_i = s_{i}'$, Bob's
view of the protocol executed with $(s_0,s_1)$ is close to his view of the
protocol executed with $(s_0',s_1')$.

The proof's strategy is to show that for $i$, $X^{\calC_{1-i}}$ has high enough 
min-entropy, given Bob's view of the protocol, in such a way that 
$Y_{1-i}$ is indistinguishable from an uniform distribution. Indistinguishability 
of Bob's views will then follow.

By the bounded storage assumption, $|\widetilde{f}(\widetilde{X})| \leq \gamma n$ with 
$\gamma < \alpha$. Then, by \lemref{lem:gr},

\begin{displaymath}
  R_{\infty}^{\varepsilon'}(X | \widetilde{f}(\widetilde{X})) \geq \alpha - \gamma -
  \frac{1+\log(1/\varepsilon')}{n} = \rho.
\end{displaymath}

Since Alice is honest, $\calA$ is randomly chosen. Lets consider a random 
subset $\widetilde{\calC}$ of $\calA$ such that  $|\widetilde{\calC}|=\ell$. This is an 
$(\mu,\nu,e^{-\ell \nu^2/2})$-averaging sampler for any $\mu,\nu >0$ according 
to \lemref{lem:random-subset}. By setting $\mu=\frac{\rho-2\tau}{\log (1/ \tau)}$, 
$\nu=\frac{\tau}{\log (1/ \tau)}$, we have by \lemref{lem:vadhan_minentropy} that 
\begin{displaymath}
R_{\infty}^{\varepsilon'+\varepsilon''}(X^{\widetilde{\calC}} | \calA, \widetilde{\calC}, \widetilde{f}(\widetilde{X})) \geq \rho - 3\tau,
\end{displaymath}
for $\varepsilon''=e^{-\ell \nu^2/2}-2^{-\Omega(\tau n)}$. For 
$\widetilde{\varepsilon}=(\varepsilon'+\varepsilon'')^{1-\zeta}$, let 
$\Bad$ be the set of $\widetilde{\calC}$'s such that $R_{\infty}(X^{\widetilde{\calC}} | \calA, \widetilde{\calC}, \widetilde{f}(\widetilde{X}))$
is not $\widetilde{\varepsilon}$-close to $(\rho - 3\tau)$-min entropy rate. Due to the above 
equation the density of $\Bad$ is at most $(\varepsilon'+\varepsilon'')^{\zeta}$. Then
the size of the set $T \subset \bits^m$ of strings that maps (using the dense encoding scheme) to subsets in $\Bad$ is at most 
$(\varepsilon'+\varepsilon'')^{\zeta}2^m \leq 2^t$. Hence the properties of the interactive 
hashing protocol guarantee that with overwhelming probability there
will be an $i$ such that
$$R_{\infty}^{\widetilde{\varepsilon}}(X^{\calC_{1-i}} | \calA, \calC_{1-i}, \widetilde{f}(\widetilde{X}), M_{IH}) \geq \rho - 3\tau,$$
where $M_{IH}$ are the messages exchanged during the interactive hashing protocol.
We now show that $X^{\calC_{1-i}}$ has high min-entropy even when given
$Z_{i},Y_{i},Q_i$. We can see $(Z_{i},Y_{i},Q_i)$ as a random variable over
$\{0,1\}^{(2 m_F + 1-\beta) \ell}$. Then, by \lemref{lem:gr},
%%%%%%%%%%%%%%%%%%%%%%%%%%%%%%%
% If changing, change both
%%%%%%%%%%%%%%%%%%%%%%%%%%%%%%%
\ifnum\breakequations=1
	\begin{eqnarray*}
	&& R_{\infty}^{\hat{\varepsilon}+\sqrt{8\widetilde{\varepsilon}}}(X^{\calC_{1-i}} | \calA, \calC_{1-i}, \widetilde{f}(\widetilde{X}), M_{IH},Z_{i},Y_{i},Q_i) 
	\geq \\ 
	&& \qquad \geq \rho + \beta - 3\tau-2 m_F - 1 - \frac{1+\log(1/\hat{\varepsilon})}{\ell} = k_F.
	\end{eqnarray*}
\else
	$$R_{\infty}^{\hat{\varepsilon}+\sqrt{8\widetilde{\varepsilon}}}(X^{\calC_{1-i}} | \calA, \calC_{1-i}, \widetilde{f}(\widetilde{X}), M_{IH},Z_{i},Y_{i},Q_i) 
	\geq \rho + \beta - 3\tau-2 m_F - 1 - \frac{1+\log(1/\hat{\varepsilon})}{\ell} = k_F.$$
\fi

Thus setting $\varepsilon'$ and $\hat{\varepsilon}$ to be negligible in $\ell$, the use of the 
$(k_F \ell,\varepsilon_F,\delta+\xi,0)$-fuzzy extractor to obtain $y_i$ that 
is used as an one-time pad guarantees that only negligible information about 
$s_{i \oplus e}$ can be leaked and that the protocol is $\lalice$-secure 
for Alice for negligible $\lalice$.
\end{IEEEproof}

\section{Discussion}
\label{sec:discussion2}

In this section we briefly discuss the protocols we obtained in terms of their robustness against noise. 

For the case of oblivious transfer, our best protocol works for levels of noise such that $h(2\delta)<\alpha-\gamma.$ Putting $\alpha=1$ and $\gamma=0.5$ (that means, a public string that is perfectly random and the bound on the memory equal to half the length of the publicly available string), we have that oblivious transfer is possible if $\delta<0.055$. Figure \ref{fig:OT} presents the maximum supported values of noise for $\alpha-\gamma$ ranging from $0$ to $1$. 

\begin{figure}
  \centering
  \includegraphics[scale=0.7]{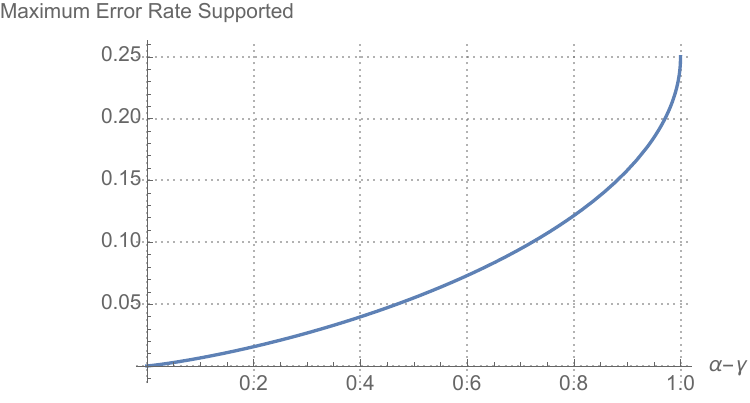}
  \caption{Acceptable levels of noise as a function of $\alpha-\gamma$  for the oblivious transfer protocol}
  \label{fig:OT}
\end{figure}

Our non-interactive commitment protocol works for $h(\delta)<(\alpha-\gamma)/2.$ For $\alpha=1$ and $\gamma=0.5$ we have that non-interactive commitments are possible in the noisy memory bounded model if $\delta < 0.041$. Figure \ref{fig:NICP} presents the maximum supported values of noise for $\alpha-\gamma$ ranging from $0$ to $1$. 

\begin{figure}
  \centering
  \includegraphics[scale=0.7]{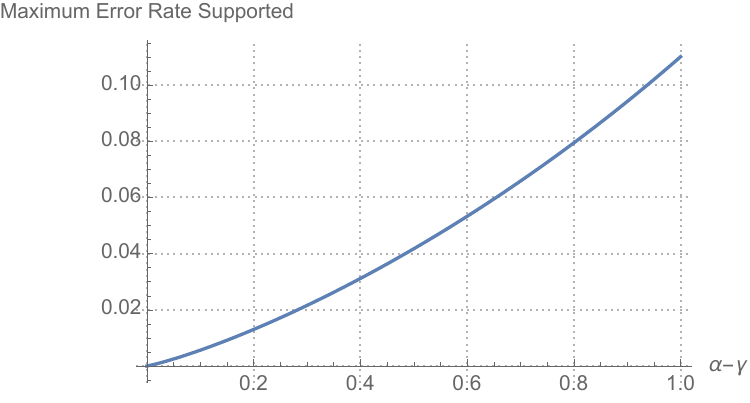}
  \caption{Acceptable levels of noise as a function of $\alpha-\gamma$  for the non-interactive commitment protocol}
  \label{fig:NICP}
\end{figure}

Finally, interactive commitments are possible if $h(\delta)<\alpha-\gamma.$ For the same settings ($\alpha=1$ and $\gamma=0.5$), this gives us a maximum noise rate of $\delta<0.11$. Figure \ref{fig:INCP} presents the maximum supported values of noise for $\alpha-\gamma$ ranging from $0$ to $1$. 

\begin{figure}
  \centering
  \includegraphics[scale=0.7]{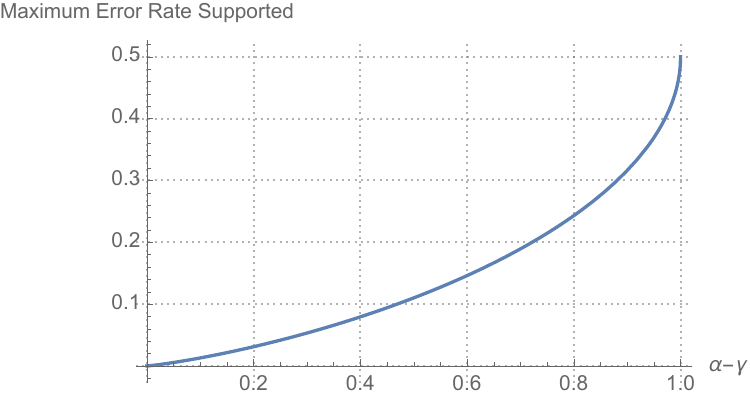}
  \caption{Acceptable levels of noise as a function of $\alpha-\gamma$ for the interactive commitment protocols}
  \label{fig:INCP}
\end{figure}

\section{Conclusion}
\label{sec:discussion}

In this work we presented the first protocols for commitment and oblivious 
transfer in the bounded storage model with errors, thus extending the previous results existing in the literature for key agreement \cite{TCC:Ding05}. 
%The proposed protocols are based on different techniques: while one of the commitment protocols is based on universal hashing and a typicality test and the other commitment protocol on an interactive hashing and a typicality test, the OT protocol is based on error-correcting information in addition to an interactive hashing. 
As expected, our protocols work for a limited range of values of the noise
parameter $\delta$. The allowed range for our commitment schemes is 
different than the one for the OT protocol. For the case of commitment schemes, the range of noise that could be tolerated depended on the round complexity of the proposed protocols: extra rounds helped tolerating a more severe noise. 

There are many open questions that follow our results here: 

\begin{itemize}
\item To prove the impossibility of commitment protocols when $h(\delta) \geq  \alpha -\gamma$. 
\item To obtain efficient OT protocols that work for the range of noise achieved by our protocols based on random linear codes. 
\item What is the best range of noise that can be achieved by non-interactive commitment protocols?  
\item Is there an intrinsic difference in the level of noise tolerated by bit commitment and OT protocols?
\end{itemize}

We do conjecture that there exists an intrinsic difference between OT and commitment schemes in the sense that there exist levels of noise so that one of them is possible but not the other. If this conjecture is proven, this would sharply contrast with the noise-free bounded memory model, where there is an all-or-nothing situation: either one has OT and bit commitment or one has nothing. Our main argument in support of this conjecture is the need for error correction in the case of OT protocols in the bounded storage model. In the case of commitment protocols error correction is not needed, the main tool used to prevent Alice from cheating is a typicality test.

\def\shortbib{0}

\subsection{Smooth Entropies and their Applications in Cryptography}\label{sec:smo}

We recall that the conditional min-entropy is defined as 
  
  \begin{align*}
    H_{\infty}(X | Y) &= \min_{y \in \mathcal{Y}} \min_{x \in \mathcal{X}}
    (-\log P_{X | Y=y}(x))
  \end{align*}
for a probability mass function $P_{XY}$ over $\mathcal{X} \times
  \mathcal{Y}$. 
  
  One faces some difficulties when trying to characterize the amount of randomness one can extract from $X$ given $Y$ based on this classical definition of conditional min-entropy. The first problem is that min-entropy can be highly sensitive (can change drastically) due to events that occur with a very small probability. For example, assume that $\varepsilon$ is a negligible parameter. Consider random variables $X$ over $\calX=\{1,\ldots,2^\ell\}$ and $Y$ over 
$\calY = \{0,1\}$. Intuitively $Y$ can express the occurrence of some bad event. Let $P_{X|Y}(x|y=0)=1/|\calX|$ for all $x \in \calX$, and $P_{X|Y}(x=1|y=1)=1$. On the one hand, if $P_{Y}(y=1)=0$, then $H_{\infty}(X|Y) = \ell$. On the other hand, if $P_{Y}(y=1)=\varepsilon$ is non-zero (even if it is very small value), then $H_{\infty}(X|Y)=0$. 

A second problem is that the usual chain rule for entropy, which holds for the conditional Shannon entropy ($H(X|Y)=H(XY)-H(X)$), does not hold anymore for the conditional min-entropy as here defined. This makes the use of the traditional conditional min-entropy quite limited in the context of randomness extraction and its applications to cryptography. The same problems are also found with the similarly defined max-entropy $H_{0}(X|Y)$. 

Smooth entropies solve these problems by ``smoothing out'' events that have a small probability. We recall the definition of conditional smooth min-entropy.  For $\varepsilon > 0$, the 
\textnormal{$\varepsilon$-smooth min-entropy of $X$ given $Y$} is 

\begin{displaymath}
H_\infty^\varepsilon(X|Y)=\max\limits_{X'Y':\|P_{X'Y'}-P_{XY}\|
\leq\varepsilon}H_\infty(X'|Y').
\end{displaymath}
for probability mass functions $P_{XY}$ and $P_{X'Y'}$. Thus, $H_\infty^\varepsilon(X|Y)$ is maximized over all probability distributions that are $\varepsilon$ distant from the original distribution $P_{XY}$. Going back to our original example where $P_{X|Y}(x|y=0)=1/|\calX|$  and $P_{Y}(y=1)=\varepsilon>0$, we have that $H_\infty^\varepsilon(X|Y)= \ell$. $H_\infty^\varepsilon(X|Y)$ is ``robust'' to any event happening with probability at most $\varepsilon$. 

The smooth min-entropy also has the nice property of having  approximate chain rules (that depend on $\varepsilon$). 

\begin{lemma}[\cite{AC:RenWol05}]
  \label{lem:chain}
  Let $\varepsilon,\varepsilon',\varepsilon'' > 0$ and $P_{XYZ}$ be a tripartite probability
  mass function. Then
%%%%%%%%%%%%%%%%%%%%%%%%%%%%%%%
% In case of changes, change both
%%%%%%%%%%%%%%%%%%%%%%%%%%%%%%%
  \ifnum\breakequations=1
  	\begin{align*}
    		H_{\infty}^{\varepsilon + \varepsilon'}(X, Y | Z) &\geq
    		H_{\infty}^{\varepsilon}(X | Y, Z) + H_{\infty}^{\varepsilon'}(Y | Z) \text{ and}\\
    		H_{\infty}^{\varepsilon}(X, Y | Z) &<
    		H_{\infty}^{\varepsilon+\varepsilon' + \varepsilon''}(X | Y, Z) +
    		H_0^{\varepsilon''}(Y | Z) \\ 
    		& \qquad + \log(1/\varepsilon'). 
  	\end{align*}
   \else
   	\begin{align*}
    		H_{\infty}^{\varepsilon + \varepsilon'}(X, Y | Z) &\geq
    		H_{\infty}^{\varepsilon}(X | Y, Z) + H_{\infty}^{\varepsilon'}(Y | Z) \text{ and}\\
    		H_{\infty}^{\varepsilon}(X, Y | Z) &<
    		H_{\infty}^{\varepsilon+\varepsilon' + \varepsilon''}(X | Y, Z) +
    		H_0^{\varepsilon''}(Y | Z) 
    		+ \log(1/\varepsilon'). 
  	\end{align*}
   \fi
\end{lemma}

It is shown in \cite{AC:RenWol05} that the conditional smooth min-entropy $H_\infty^\varepsilon(X|Y)$ optimally characterizes the amount of randomness that one can extract from $X$ when $Y$ is given as side information.  

The $\varepsilon$-smooth variants  directly allow to capture bad events that happens with probability at most $\varepsilon$. Another possible alternative would be to use ``average min-entropy'', in which the min-entropy is averaged over the random variable $Y$ \cite{EC:DodReySmi04}. In this way, small probability events are naturally taken care of in the average computation. Moreover, one can also show that average min-entropy also satisfies an approximate chain rule \cite{EC:DodReySmi04}. Smooth min-entropy and average min-entropy are asymptotically identical when applied to independent and identically distributed distributions.

\end{document}